\documentclass[sigconf]{acmart}
\usepackage{tikz}
\usetikzlibrary{decorations.pathreplacing}
\usepackage{bbding}
\usepackage{url}
\usepackage{amsmath}
\usepackage{nicefrac}
\usepackage{amsthm}

\usepackage{amssymb}
\usepackage{algorithm}
\usepackage{algorithmic}
\usepackage{color}
\usepackage{colortbl,booktabs}
\usepackage{multirow}
\usepackage{subfigure}
\usepackage{xcolor}
\newtheorem{myDef}{Definition}
\newtheorem{theorem}{Theorem}

\newcommand{\tabincell}[2]{
\begin{tabular}{@{}#1@{}}#2\end{tabular}
}
\usepackage{filecontents}

\AtBeginDocument{%
\providecommand\BibTeX{{%
\normalfont B\kern-0.5em{\scshape i\kern-0.25em b}\kern-0.8em\TeX}}}

\setcopyright{acmcopyright}
\copyrightyear{2021}
\acmYear{2021}
\acmDOI{10.1145/1122445.1122456}

\acmConference[Woodstock '18]{Woodstock '18: ACM Symposium on Neural
Gaze Detection}{June 03--05, 2018}{Woodstock, NY}
\acmBooktitle{Woodstock '18: ACM Symposium on Neural Gaze Detection,
June 03--05, 2018, Woodstock, NY}
\acmPrice{15.00}
\acmISBN{978-1-4503-XXXX-X/18/06}

\begin{document}

\title{Merkle-Sign: Watermarking Framework for Deep Neural Networks in Federated Learning}

\author{Fang-Qi Li}
\email{solour_lfq@sjtu.edu.cn}
\affiliation{%
\institution{Shanghai Jiao Tong University, School of Electronic Information and Electrical Engineering}
\city{Shanghai}
\country{China}
}

\author{Shi-Lin Wang}
\email{wsl@sjtu.edu.cn}
\authornotemark[1]
\affiliation{%
\institution{Shanghai Jiao Tong University, School of Electronic Information and Electrical Engineering}
\city{Shanghai}
\country{China}
}

\author{Alan Wee-Chung Liew}
\email{a.liew@griffith.edu.au}
\authornotemark[1]
\affiliation{%
\institution{Griffith University, School of Information and Communication Technology}
\city{Gold Coast Campus}
\country{Australia}
}


\begin{abstract}
With the wide application of deep learning models, it is important to verify an author's possession over a deep neural network model, e.g. by embedding watermarks and protect the model.
The development of distributed learning paradigms such as federated learning (FL) leverages new challenges for model protection.
Apart from the independent verification of all authors' ownership, the author collaboration should be able to recover its participant's identity against adaptive attacks and trace traitors.
To meet those requirements, in this paper, we demonstrate a watermarking framework, $\texttt{Merkle-Sign}$, for protecting deep neural networks in FL.
By incorporating state-of-the-art watermarking schemes and the cryptological primitive designed for distributed storage, the framework meets the prerequisites for ownership verification in FL.
This work paves the way for generalizing watermark as a practical security mechanism for protecting deep learning models in distributed learning platforms.
\end{abstract}

\ccsdesc[500]{Security and privacy~Digital rights management}
\ccsdesc[500]{Security and privacy~Domain-specific security and privacy architectures}

\keywords{deep learning model protection, federated learning, machine learning security}

\maketitle

\section{Introduction}
Deep neural networks (DNN) are intelligent systems that provide services by learning from data and consuming enormous computational resources.
DNN models have been successfully applied in internet of things~\cite{iot2,iot3}, intrusion detection~\cite{shone2018deep,yin2017deep}, etc.
With more emerging applications, their reliability and security are gaining more attention.
A crucial task in artificial intelligence security is to protect DNN models as intellectual properties by proving authors' ownership.
Mechanisms designed for the integrity of traditional communication are inappropriate for DNN for two reasons:
(1) Unlike traditional communication where information is passed from one agent to another, DNN models are usually uploaded onto public websites for credit, visibility, and peer review.
So schemes designed for identity verification in bidirectional channels are unsuitable for DNN models.
(2) It is easy to modify a published model (e.g., changing the last bits of one parameter) without significantly damaging its commercial utility.
This fact challenges the direct application of cryptological primitives under which changing a single bit within the protected object results in a completely different signature.

Having observed that DNN models share similar properties with multi-media objects (broadcast transmission and semantic invariance under slight modifications), researchers resorted to watermark, which protects the integrity for multi-media objects.
The metrics for evaluating a DNN watermarking scheme are similar to those for the multi-media setting.
Different kinds of redundancy within DNN models, such as parameter representation, outputs of intermediate layers, and backdoor triggers can be exploited to encode the author's identity as watermarks.
There have been various watermarking schemes designed to satisfy these security requirements.

Modern paradigms of model training and distributing such as Federated Learning (FL)~\cite{li2020federated} raise new challenges for DNN watermarking.
In FL, authors sharing different local datasets collaborate to train a model.
As a result, the final product contains the contribution of all the authors.
Such an organization results in a different threat model and ownership verification requirements from those in the traditional setting where one party completes the training of the entire DNN model.

In FL, a watermarking scheme should encode all authors' identity information into the model while:
(1) Each agent participating in the training can prove its ownership without informing its co-authors or the aggregator server.
(2) One agent participating in the training cannot falsify itself as its co-author.
(3) If one agent undergoes several adaptive attacks so its watermark is erased then its co-authors can recover its identity proof.
(4) If a traitor within the author union pirates the model then its identity can be correctly tracked.
It is hard to fulfill all these requirements without a sophisticated verification protocol.
Meanwhile, such protocol would raise additional demands for the underlying watermarking scheme for identity encoding.
These concerns increase the difficulty in model protection for FL.

Being confronted with all these challenges, we propose a unified DNN watermarking framework that meets the practical security requirements for FL.
The proposed framework adopts a watermarking scheme and a public verification protocol to ensure unforgeable and robust ownership proof.
We prove that such a scheme meets the requirements for FL listed above.
Apart from the top-level protocol, our design puts forward more requirements for the underlying watermarking schemes.
To the best of our knowledge, this is the first proposal for provable ownership protection in the FL scenario.
The contributions of this paper are:
\begin{itemize}
\item We formulate security requirements for FL, a practical scenario for secure machine learning.
\item A framework for ownership verification, $\texttt{Merkle-Sign}$, is proposed to meet all the requirements for FL.
We also propose an underlying watermarking scheme, $\texttt{ATGF}$, to increase the applicability of $\texttt{Merkle-Sign}$ in real-world settings.
\item Experiments demonstrate the utility of the proposed framework and examine the capability of established watermarking schemes in FL.
\end{itemize}

\section{Backgrounds and Preliminaries}
\label{section:2}
\subsection{DNN watermark}
Let $M_{\text{clean}}$ be a model trained to fulfil a primary task $\mathcal{T}_{\text{primary}}$.
The author embeds its identity information, $\texttt{key}\in\mathcal{K}$, into the model, where $\mathcal{K}$ is the space of legal keys.
Such embedding can be revealed later by a module $\texttt{verify}$, a probabilistic algorithm with binary output.
A watermarking scheme $\texttt{WM}=\left\{\texttt{Gen},\texttt{Embed} \right\}$ consists of one module for key generation and one for key embedding:
$$
\begin{aligned}
\texttt{key}&\leftarrow \texttt{Gen}(1^{N}),\\
(M_{\text{WM}},\texttt{verify})&\leftarrow \texttt{Embed}(M_{\text{clean}},\texttt{key}),
\end{aligned}
$$
where $N$ is the security parameter.
When the author finds that its model has been stolen, it provides $\left\{\texttt{key},\texttt{verify}\right\}$ as its evidence.
If $\texttt{verify}$ agrees with $\texttt{key}$ concerning the suspicious model $M$ then the author's ownership is proven.
Such proof can be conducted and announced by a trusted third party.
This centralized setting introduces extra security risks and a tremendous burden to the center.
Instead, as~\cite{mengelkamp2018blockchain, reyna2018blockchain, khan2018iot}, the proof as a consensus service can be done in a public and decentralized manner so any party can participate in it.
In particular, we assume that there exists a community of public verification agents.
Each individual within the community examines the watermark independently and provides its result.
The ownership verification outcome is voted across the entire community to defend against potential conspiracy and perjury.
This community is motivated by the necessity of protecting DNN intellectual property.
An agent correctly participating in verification is going to be assigned corresponding credits with which it can initialize a verification process for its product.

Current \emph{watermarking schemes} mainly focus on embedding information into the DNN models.
Zhang \textit{et al.} proposed a backdoor-based watermarking scheme~\cite{zhang2018protecting} for the black-box setting using triggers.
Backdoor is an attack against DNN by triggering unexpected results.
In backdoor-based watermarking, an author conducts the backdoor attack on its model and claims the exceptional output as its identity proof.
Schemes ensuring the secure generation and transmission of backdoors provide better security~\cite{adi2018turning,zhu2020secure}.
Moreover, variational autoencoder~\cite{li2019prove}, adversarial samples~\cite{le2020adversarial} and out-of-range samples~\cite{li2019persistent} in $\texttt{Wonder Filter}$, have been adopted as candidates for the backdoor.

Uchida \textit{et al.}~\cite{uchida2017embedding} proposed the pioneering white-box watermarking scheme.
Their scheme embeds the author's digital signature into the model's weight.
In~\cite{guan2020reversible}, Guan \textit{et al.} inserted a reversible watermark into the model's parameter.
\texttt{Deep-Sign} proposed by Davish \textit{et al.}~\cite{darvish2019deepsigns} builds the watermark on the intermedium output of the DNN model on certain samples instead of the weights.
The scheme \texttt{MTL-Sign} proposed by Li~\textit{et al.} in~\cite{ours} models watermarking embedding as an extra learning task.
Since regulating the intermediate output of the DNN model is more difficult than tuning weights, these hybrid schemes are more robust and secure than weight-based schemes.
Zhang \textit{et al.}~\cite{zhang2021deep} adopted deep watermarking to protect DNN models specialized for image processing.

Despite all the recent works, DNN watermarking is still far from practical application.
Most of the proposals overlooked the top-level protocol for key generation, key distribution, and verification.
A practical \emph{watermarking framework} must include both the bottom-level watermarking scheme and the top-level protocol.
The authorship verification of DNN models should be conducted publicly in a decentralized manner.
However, the public verification process itself raises additional risks as we will elaborate on later.


\begin{figure*}[!htbp]
\centering
\subfigure[The first verification.]{
\begin{minipage}[htbp]{0.3\linewidth}
\centering
\includegraphics[width=5.5cm]{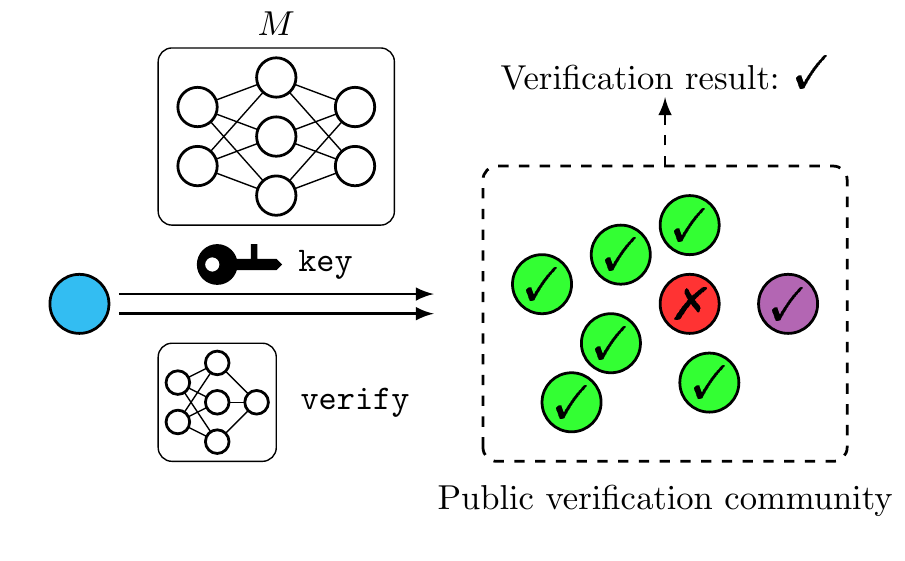}
\end{minipage}%
}%
\subfigure[The spoil attack.]{
\begin{minipage}[htbp]{0.3\linewidth}
\centering
\includegraphics[width=5cm]{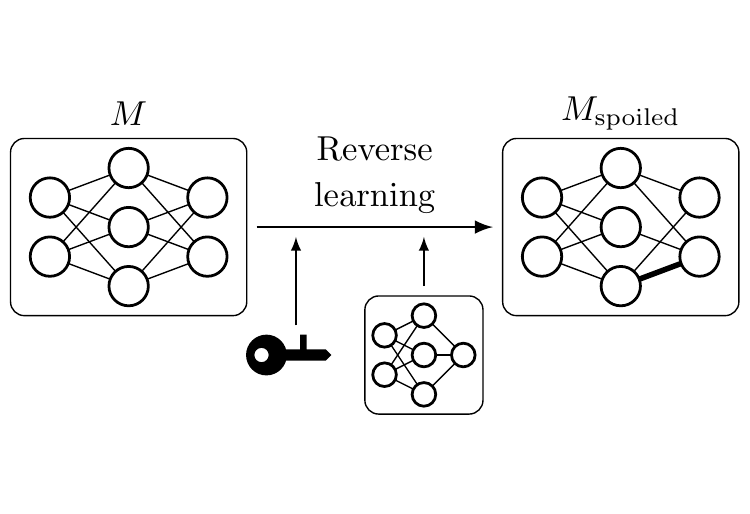}
\end{minipage}%
}%
\subfigure[The second verification.]{
\begin{minipage}[htbp]{0.3\linewidth}
\centering
\includegraphics[width=5.5cm]{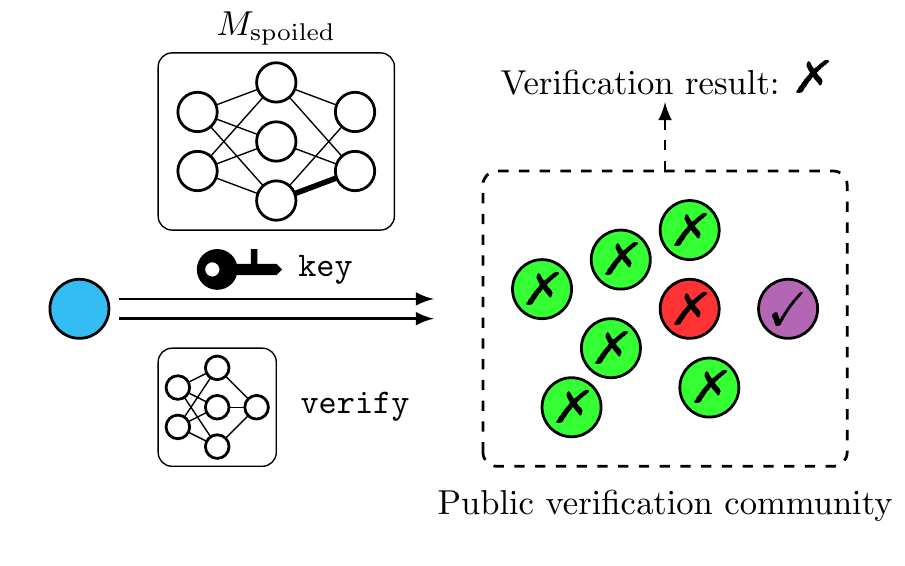}
\end{minipage}%
}%
\caption{The blue node is the author, green nodes are benign agents, the red node is a malicious agent, and the purple one is the eavesdropping adversary.}
\label{figure:1}
\end{figure*}

\begin{table*}[htbp]
\caption{Watermarking schemes evaluated by the basic security requirements.}
\begin{center}
\begin{tabular}{cc|c|c|c|c|c|c|c}
\toprule[1.5pt]
\multicolumn{2}{c|}{\multirow{2}{*}{\tabincell{c}{\ \\\textbf{Watermarking}\\ \textbf{schemes}}}} & \multicolumn{7}{c}{\textbf{Security requirements}}\\
\cline{3-9}
& & Correctness & \tabincell{c}{Functionality-\\preserving.} & \tabincell{c}{Security\\against\\tuning.} & Covertness. & Privacy. & \tabincell{c}{Security\\against\\overwriting.} & \tabincell{c}{Security\\against\\piracy.} \\
\midrule[1pt]
\multirow{4}{*}{\tabincell{c}{Backdoor-\\based.}} & \multicolumn{1}{|c|}{Zhang~\emph{et al.}~\cite{zhang2018protecting}} & \checkmark & \checkmark & \checkmark & \checkmark & $\times$ & $\times$ & $\times$ \\
\cline{2-9}
& \multicolumn{1}{|c|}{\texttt{Wonder Filter}~\cite{li2019persistent}} & \checkmark & \checkmark & \checkmark & $\times$ & $\times$ & $\times$ & \checkmark \\
\cline{2-9}
& \multicolumn{1}{|c|}{Adi~\emph{et al.}~\cite{adi2018turning}} & \checkmark & \checkmark & \checkmark & \checkmark & $\times$ & \checkmark & \checkmark \\
\midrule
\multirow{2}{*}{\tabincell{c}{Weight-\\based.}} & \multicolumn{1}{|c|}{Uchida~\emph{et al.}~\cite{uchida2017embedding}} & \checkmark & \checkmark & \checkmark & \checkmark & $\times$ & \checkmark & \checkmark \\
\cline{2-9}
& \multicolumn{1}{|c|}{Guan~\emph{et al.}~\cite{guan2020reversible}} & \checkmark & \checkmark & $\times$ & $\times$ & $\times$ & $\times$ & \checkmark \\
\midrule
\multirow{2}{*}{\tabincell{c}{Hybrid\\ schemes.}} & \multicolumn{1}{|c|}{\texttt{Deep-Sign}~\cite{darvish2019deepsigns}} & \checkmark & \checkmark & \checkmark & $\times$ & $\times$ & \checkmark & \checkmark \\
\cline{2-9}
& \multicolumn{1}{|c|}{\texttt{MTL-Sign}~\cite{ours}} & \checkmark & \checkmark & \checkmark & \checkmark & \checkmark & \checkmark & \checkmark \\
\bottomrule[1pt]
\multicolumn{9}{l}{\checkmark means the security requirement holds. $\times$ means the corrseponding requirement is not met/formally proven.}\\
\end{tabular}
\label{table:1}
\end{center}
\end{table*}

A DNN watermarking scheme $\texttt{WM}$ has to satisfy the following basic security requirements.

\subsubsection{Correctness}
The author can prove its ownership correctly:
$$
\text{Pr}\left\{\texttt{verify}(M_{\text{WM}},\texttt{key})=1 \right\}\geq 1-\epsilon,
$$
where $\epsilon$ is a negligible function of the security parameter $N$.
For an adversary with the incorrect evidence, $\texttt{verify}$ is not going to mistake it as the authentic author:
$$
\text{Pr}_{\texttt{key}'\leftarrow\mathcal{K}}\left\{\texttt{verify}(M_{\text{WM}},\texttt{key}')=1 \right\}\leq\epsilon,
$$
in which $\texttt{key}'$ is a uniformly sampled instance from $\mathcal{K}$.

\subsubsection{Functionality-preserving}
Let $\mathcal{E}$ denote a task-dependent metric that evaluates a model's performance, e.g., the accuracy on the test set for classification.
It is desirable that watermarking does not substantially decrease the model's functionality, so:
$$\mathcal{E}(M_{\text{WM}})\geq \mathcal{E}(M_{\text{clean}})-\epsilon$$
holds almost always, where $\epsilon$ measures the author's tolerance on the watermark's side-effect.

\subsubsection{Security against tuning}
The adversary can fine-tune a model on a smaller dataset similar to that of $\mathcal{T}_{\text{primary}}$'s.
Such tuning could ruin some weight-based watermarks like the reversible system in~\cite{guan2020reversible}.
A watermarking scheme is secure against tuning if tuning cannot invalidate the ownership verification, which can be formulated as:
$$
\text{Pr}\left\{\texttt{verify}(M_{\text{tuned}},\texttt{key})=1 \right\}\geq 1-\epsilon,
$$
where $M_{\text{tuned}}$ is obtained by having an adversary tune $M_{\text{WM}}$.

\subsubsection{Covertness}
The watermark should be hidden from the adversary.
A necessary condition for the watermark's covertness is: the parameters of $M_{\text{WM}}$ should not deviate from those of $M_{\text{clean}}$ too much.
This deviation can be measured by the statistics of the two models' parameters as in the \emph{property inference}~\cite{wang2019robust}.

\subsubsection{Privacy-preserving}
In scenarios where anonymity is necessary such as model competitions, the watermarked model must not reveal the author's identity.
A watermarking scheme is privacy-preserving~\cite{ours} if no efficient adversary, given $\texttt{key}_{1}$ and $\texttt{key}_{2}$, can distinguish $M_{\text{WM},1}$ from $M_{\text{WM},2}$ where:
$$
\begin{aligned}
(M_{\text{WM},1},\texttt{verify}_{1})&\leftarrow \texttt{Embed}(M_{\text{clean}},\texttt{key}_{1}).\\
(M_{\text{WM},2},\texttt{verify}_{2})&\leftarrow \texttt{Embed}(M_{\text{clean}},\texttt{key}_{2}).
\end{aligned}
$$

\subsubsection{Security against overwriting}
An adversary can embed its watermark into the model and \emph{redeclare} the model with its watermark as its product to the verification community.
Such overwriting must not invalidate the original watermark.
The redeclaration can be solved by either using an authorized time server or a decentralized consensus protocol with which an author authorizes its time-stamp to the verification community before publishing the DNN model~\cite{ours}.

\subsubsection{Security against piracy}
For watermarking schemes with structural flaws or whose key space is too small (as in the cases of most backdoor-based schemes), an adversary can trivially claim the ownership by forging the evidence.
A necessary condition for the security against ownership piracy is to incorporate cryptological primitives as in~\cite{zhu2020secure}.

\subsubsection{Security against the spoil attack}
An author might have to prove its possession over a DNN model for \emph{multiple} times.
For example, it discovers that its model has been stolen, tuned, and sold to multiple parties.
If an adversary having obtained the evidence can spoil the watermark and invalidate the verification afterward then the watermarking scheme is insecure.
We define this attack as \emph{the spoil attack}.
The procedure of a spoil attack is illustrated in Fig.~\ref{figure:1}.
An eavesdropping adversary firstly participates in the proof for a author, during which process it obtains $\left\{\texttt{key},\texttt{verify}\right\}$ for model $M$.
It then tunes $M$ into $M_{\text{spoiled}}$ in a reverse learning manner so that the original $\texttt{verify}$ cannot recognize $(M_{\text{spoil}},\texttt{key})$.
Finally, the adversary sells $M_{\text{spoiled}}$.

To the best of our knowledge, no existing watermarking scheme is secure against the spoil attack.
For backdoor-based watermarking schemes, the adversary can fit the model to randomly shuffled labels on the backdoor triggers.
For weight-based watermarking schemes, the adversary can replace the watermarked layers.
The concurrent watermarking schemes and their security level w.r.t. all requirements are listed as in Table.~\ref{table:1}

\subsection{Federated learning}
To associate a broader range of data sources while preserving the privacy of each data provider, many distributed learning frameworks have been proposed.
The representative example is Federated Learning (FL)~\cite{li2020federated}, where a collection of $K$ participants $\mathcal{U}=\left\{u_{k}\right\}_{k=1}^{K}$ and optionally an aggregator server $\mathcal{A}$ cooperate to train a model.
FL can be classified into horizontal FL and vertical FL.
In horizontal FL, each party holds different samples with an identical set of features.
This paradigm has been applied in medical and financial settings~\cite{yang2019federated}.
In the vertical FL, the feature spaces of different parties are different.
As illustrated in Fig.~\ref{figure:2}, the centralized client-server setting of FL consists of many epochs, during the $i$-th epoch the aggregator distributes the current model $M^{(i)}$ to all authors.
Each author $u_{k}$ independently computes the direction to which the model should evolve w.r.t. its local dataset $\mathcal{D}_{k}$ and transmits the gradient $\Delta M^{(i)}_{k}$ back.
Finally, the aggregator collects the feedback from all authors, updates the model, then starts another round.
Possible ways of model aggregation are model average, gradient average as in~\eqref{equation:ga}, secure aggregation, homomorphic encryption-based weighted average, etc~\cite{bonawitz2017practical}.
\begin{equation}
\label{equation:ga}
M^{(i+1)}=M^{(i)}+\alpha\cdot\sum_{k=1}^{K}|\mathcal{D}_{k}|\cdot \Delta M^{(i)}_{k}.
\end{equation}
During the entire procedure, each author communicates with the aggregator through an encrypted channel: it downloads the current model and uploads the information for model updating.
No data is transmitted through any public channel.

\begin{figure}[htbp]
\centering
\includegraphics[width=7.2cm]{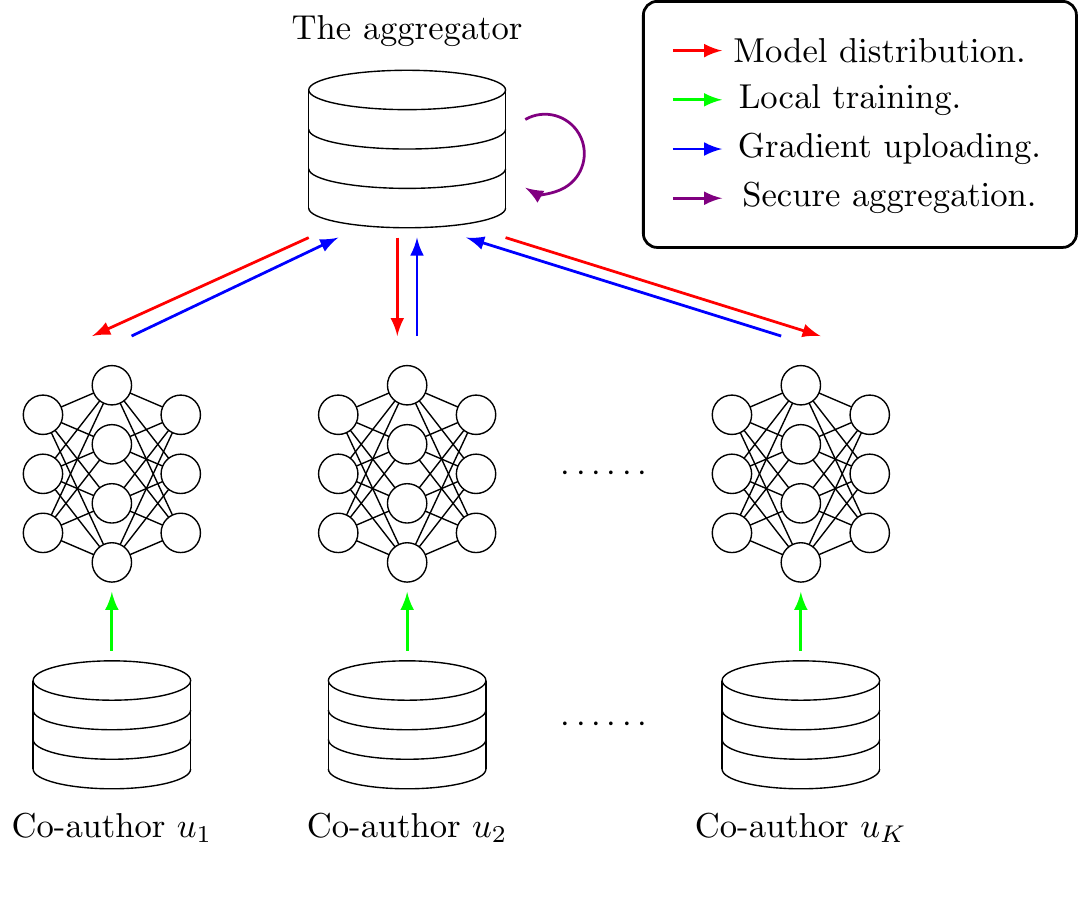}
\caption{The client-server architecture for FL.}
\label{figure:2}
\end{figure}

\begin{table*}[htbp]
\caption{Dependence of the advanced requirements on basic security requirements.}
\begin{center}
\begin{tabular}{c|c|c|c|c|c|c|c|c}
\toprule[1.5pt]
\multicolumn{1}{c|}{\multirow{2}{*}{\tabincell{c}{\ \\\textbf{Advanced}\\ \textbf{requirements}}}} & \multicolumn{8}{c}{\textbf{Basic security requirements}}\\
\cline{2-9}
& Correctness & \tabincell{c}{Functionality-\\preserving.} & \tabincell{c}{Security\\against\\tuning.} & Covertness. & Privacy. & \tabincell{c}{Security\\against\\overwriting.} & \tabincell{c}{Security\\against\\piracy.} & \tabincell{c}{Multiple-time\\verification.} \\
\midrule[1pt]
\multicolumn{1}{c|}{Independency.} & \checkmark & \checkmark & $\times$ & $\times$ & $\times$ & \checkmark & \checkmark & \checkmark \\
\midrule
\multicolumn{1}{c|}{Privacy-preserving.} & $\times$ & $\times$ & $\times$ & \checkmark & $\times$ & $\times$ & $\times$ & $\times$ \\
\midrule
\multicolumn{1}{c|}{Recovery.} & \checkmark & \checkmark & $\times$ & $\times$ & $\times$ & $\times$ & $\times$ & $\times$ \\
\midrule
\multicolumn{1}{c|}{Traitor-tracing.} & \checkmark & \checkmark & \checkmark & \checkmark & $\times$ & $\times$ & \checkmark & \checkmark \\
\bottomrule[1pt]
\multicolumn{9}{l}{\checkmark means relevant. $\times$ means irrelevant.}\\
\end{tabular}
\label{table:2}
\end{center}
\end{table*}

Current studies on the security of FL focus on privacy protection against curious or malicious aggregators or authors.
It has been claimed that the gradient information $\Delta M^{(i)}_{k}$ might leak information of $\mathcal{D}_{k}$ under reverse engineering~\cite{augasta2012reverse}.
Therefore differential privacy~\cite{wei2020federated,xu2019verifynet} methods have been widely adopted to protect the encrypted gradients.
For example, by perturbating the encrypted gradients, neither the curious aggregator nor its conspirators can deduce the noiseless gradient.
Yet the perturbation strategy keeps the summation of gradients invariant and preserves the stability in convergence.

As a separate aspect for FL security, DNN model protection is confronted with distinct threat models.

\section{The Threat Model and Security Requirements}
\label{section:3}
\subsection{Watermark capacity}
\label{section:overwriting}
It is necessary that all authors' identities are embedded into the DNN model, inclusively and independently.
Therefore practical watermarking scheme for DNN models in the FL setting has to embed multiple pieces of identity information into the model.
It is generally assumed that watermark is embedded into a DNN by utilizing the model's redundancy in parameterization.
To provide correct and independent verification for each author, it is necessary to quantitively measure this redundancy.
We define a DNN model's \emph{watermark capacity} w.r.t. $\texttt{WM}$ and the decline of its performance $\delta$, $\texttt{cap}_{\texttt{WM}}^{\delta}$ as follows:
\begin{myDef}
\emph{The watermark capacity for a DNN model, $\texttt{cap}_{\texttt{WM}}^{\delta}$, is the maximal number of keys that can be correctly embedded by $\texttt{WM}$ into the model until the DNN model's performance drops by $\delta$ w.r.t. the metric $\mathcal{E}$ defined in its primary task.}
\end{myDef}
The value $\texttt{cap}_{\texttt{WM}}^{\delta}$ measures the upper bound of the number of successfully embedded watermarks inside a model.
Formally, $\texttt{cap}_{\texttt{WM}}^{\delta}$ is the maximal $q$ satisfying the following conditions:
\begin{equation}
\label{equation:capacity}
\begin{aligned}
(M_{1},\texttt{verify}_{1})&\leftarrow\texttt{Embed}(M_{\text{clean}},\texttt{key}_{1}),\\
(M_{2},\texttt{verify}_{2})&\leftarrow\texttt{Embed}(M_{1},\texttt{key}_{2}),\\
&\cdots\\
(M_{q},\texttt{verify}_{q})&\leftarrow\texttt{Embed}(M_{q-1},\texttt{key}_{q}),\\
\mathcal{E}(M_{q})&\geq \mathcal{E}(M_{\text{clean}})-\delta,
\end{aligned}
\end{equation}
where $\texttt{key}_{q}$ is generated by $\texttt{Gen}$ and all $q$ watermarks can be correctly verified.
A very deep DNN model naturally has a large capacity since the watermarking embedding process usually modifies only a small number of parameters.
Although it is difficult to formally calculate $\texttt{cap}_{\texttt{WM}}^{\delta}$ for arbitary $\texttt{WM}$, it can be empirically estimated from~\eqref{equation:capacity}.

\subsection{Requirements in FL ownership verification}
Apart from the ordinary threats against the DNN watermark, the collaboration of authors in FL gives rise to new threats and requirements.
For example, the authors must not breach each other's privacy from the watermark.
Meanwhile, we expect that authors can help each other to recover the identity proof against the spoil attack.
In cases a malicious author~\cite{fung2018mitigating} sells the intermediate model as its product, it should be correctly identified.
The security requirements for FL are detailed as the following aspects:

\subsubsection{Independency}
When FL terminates, each author $u_{k}$ holds an evidence pair $(\texttt{key}_{k},\texttt{verify}_{k})$, in which $\texttt{key}_{k}$ is deteremined by $u_{k}$.
After the model $M$ absorbing the contributions of all authors is published, it is necessary that:
$$\text{Pr}\left\{\texttt{verify}_{k}(M,\texttt{key}_{k})=1 \right\}\geq 1-\epsilon.$$
So $u_{k}$ can independently prove its contribution in $M$ without informing any other parties.

\subsubsection{Privacy-preserving}
For author $u_{s}\in\mathcal{U}$, $u_{s}\neq u_{k}$, no information about $u_{k}$'s identity, especially $\texttt{key}_{k}$, appears in $u_{s}$'s perspective when interacting with the aggregator/being chained up in the decentralized setting.
Formally, $u_{s}$ cannot succeed in pretending to be $u_{k}$.

\subsubsection{Recovery}
Once $u_{k}$ proves its ownership over $M$, an eavesdropping adversary might conduct the spoil attack to invalidate this proof.
Given the shared interest and cooperation in property protection from all collaborating authors, it is expected that $u_{k}$'s co-authors can recover $u_{k}$'s ownership over $M$.
Hence a spoil attack against $u_{k}$'s watermark is insufficient to disprove $u_{k}$'s authorship completely.

\subsubsection{Traitor-tracing}
In DNN model commercialization, tracing of unauthorized reselling can be achieved by watermarking~\cite{9359144}.
In FL, a traitor may participate in a few epochs of training, obtain the intermediate model, then claim the current model as its product to the public, or sell the model to another party.
The aggregator supervising the FL training process, or the collaboration of honest authors in the decentralized setting, should be able to correctly identify the traitor.

The dependency of these four advanced requirements on the basic security properties is presented in Table.~\ref{table:2}.

\subsection{Requirements for the watermarking scheme in FL}
\label{section:2.3}
In FL, the identity information of all authors has to be embedded into the final model.
Moreover, the model distributed by the server at each epoch must be slightly different.
Otherwise, traitor-tracing is impossible.
Therefore, the underlying watermarking scheme should satisfy the following properties:
\begin{itemize}
\item The watermark capacity is large to enable both multiple-time and independent verification for each author.
\item The watermark embedding is efficient so that leaving hooks in DNN models for traitor-tracing is feasible.
\end{itemize}
In addition, the verification protocol should establish a correlation between the identity information of all authors to accomplish the recovery property.

\section{The Merkle-Sign Framework}
\label{section:4}
\subsection{Motivation}
To meet the requirements in FL model protection, we have to use a watermarking scheme satisfying the two requirements listed in Section~\ref{section:2.3}.
To reduce the cost of public verification, we propose a public verification protocol that takes the merit of the Merkle-tree structure.
Merkle-tree is a data structure for data integrity verification in cloud storage~\cite{li2013efficient}.
Unlike a naive hash function that simply hashes a list of files into a string, the Merkle-tree hash allows partial authentication of a subset of the input, which is necessary for distributed storage and retrieval.
This property forms the basis for independent verification and recovery.

\subsection{$\texttt{Merkle-Sign}$ for FL}
We propose \texttt{Merkle-Sign}, an efficient and secure watermarking framework for FL model protection, which operates as Algo.~\ref{algorithm:centralized} and is visualized in Figure~\ref{figure:5}.
\begin{figure}[htbp]
\centering
\includegraphics[width=8cm]{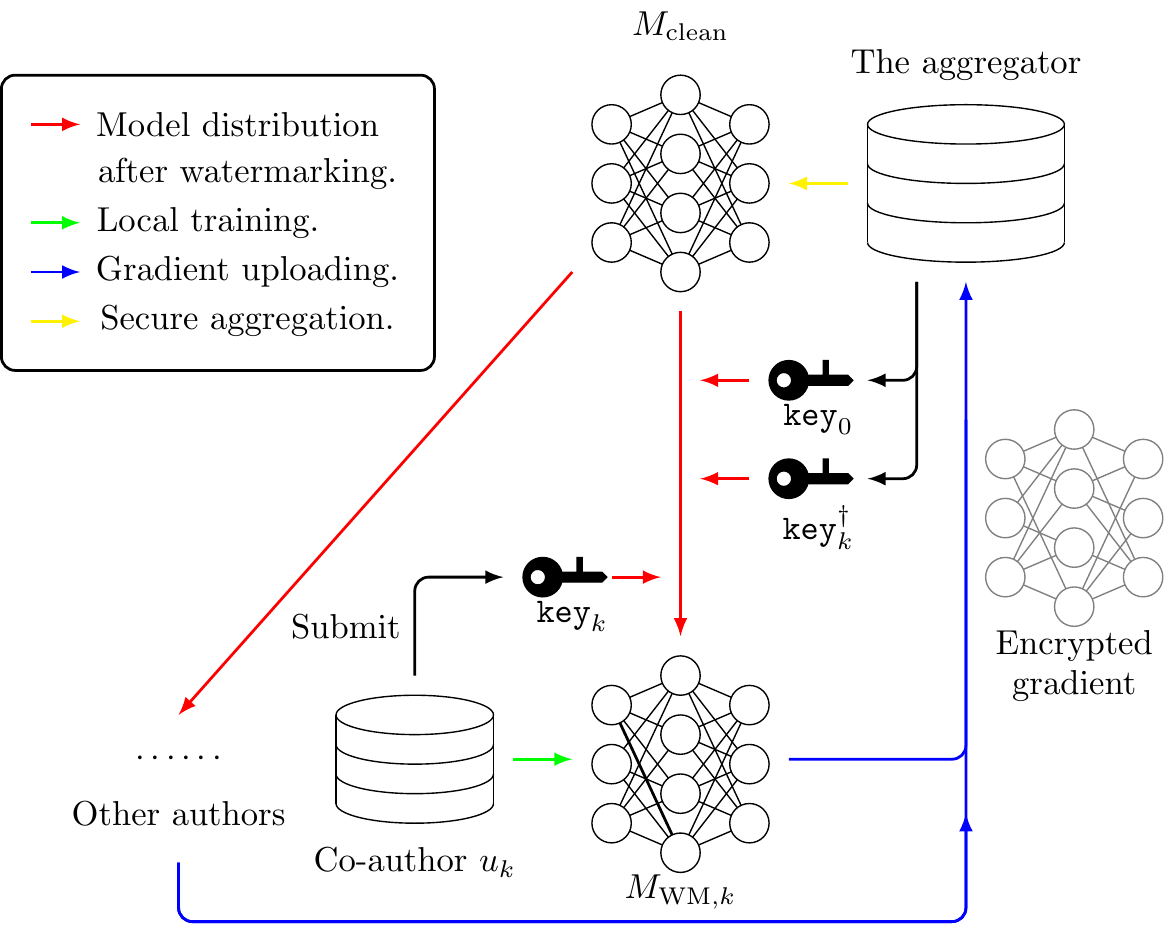}
\caption{The $\texttt{Merkle-Sign}$ watermarking framework for centralized FL.}
\label{figure:5}
\end{figure}
In which $\texttt{Merkle}$ is a Merkle-tree style combinator.
The idea is that an aggregator $\mathcal{A}$ is responsible for embedding the identity information of all authors into the final DNN model.
To ensure that all evidence is valid and the decrease of the model's performance is upper bounded by $\delta$, it is necessary that:
$$\texttt{cap}^{\delta}_{\texttt{WM}}\geq (K+1).$$
To trace the traitor before the training terminates, the intermediate model being sent to $u_{k}$ contains a hook, which is a surveillance key specialized for $u_{k}$, $\texttt{key}_{k}^{\dag}$.
This key is only known to $\mathcal{A}$, from which it can identify the traitor.

\begin{algorithm}[htbp]
\caption{$\texttt{Merkle-Sign}$ framework for centralized FL.}
\label{algorithm:centralized}
\begin{algorithmic}[1]
\REQUIRE A one-time watermarking scheme $\texttt{WM}$, security parameters $N$ and $L$, hash functions $\texttt{hash}_{1}$ and $\texttt{hash}_{2}$.
\ENSURE A watermarked model $M_{\mathcal{A}}$ and evidence.
\STATE $\mathcal{A}$ generates $\texttt{key}_{0}\leftarrow\texttt{Gen}(1^{N})$.
\STATE Each $u_{k}\in\mathcal{U}$ generates $\texttt{key}_{k}\leftarrow\texttt{Gen}(1^{N})$, submits it and $\texttt{hash}_{1}(\texttt{key}_{k}|u_{k}) $ to $\mathcal{A}$.
\STATE For each $u_{k}$, $\mathcal{A}$ generates $\texttt{key}_{k}^{\dag}\leftarrow\texttt{Gen}(1^{N})$.
\STATE $\mathcal{A}$ initializes a clean model.
\WHILE {not terminate}
\STATE For each $u_{k}$, $\mathcal{A}$ embeds $\texttt{KEYs}_{k}=(\texttt{key}_{0}, \texttt{key}_{k}^{\dag})$ into the current model as $M_{\text{WM},k}$, obtains $\texttt{VERs}_{k}=(\texttt{verify}_{0},\texttt{verify}_{k}^{\dag})$.
\STATE $\mathcal{A}$ signs and broadcasts:
$$\langle \texttt{time}\|\texttt{Merkle}(\texttt{KEYs}_{k},\texttt{VERs}_{k},\texttt{info}) \rangle.$$
\STATE $\mathcal{A}$ transmits $M_{\text{WM},k}$ to $u_{k}$.
\STATE Each $u_{k}$ uploads the encrypted gradient to $\mathcal{A}$.
\STATE $\mathcal{A}$ averages the gradients and updates the clean model.
\ENDWHILE
\STATE $\mathcal{A}$ embeds $\texttt{KEYs}=\left\{\texttt{key}_{k} \right\}_{k=0}^{K}$ into the final model $M_{\mathcal{A}}$, obtains $\texttt{VERs}=\left\{\texttt{verify}_{k}\right\}_{k=0}^{K}$.
\STATE $\mathcal{A}$ signs and broadcasts:
$$\langle \texttt{time}\|\texttt{Merkle}(\texttt{KEYs},\texttt{VERs},\texttt{info}) \rangle.$$
\STATE $\mathcal{A}$ sends the intermedia of $\texttt{Merkle}(\texttt{KEYs},\texttt{VERs},\texttt{info})$ and $\texttt{verify}_{k}$ to $u_{k}$.
\STATE $\mathcal{A}$ publishes $M_{\mathcal{A}}$.
\end{algorithmic}
\label{exp:1}
\end{algorithm}

To reduce the communication traffic between the FL participants and the verification community while enabling the independence and recovery property, we adopt the Merkle-tree combinator.
A hash function $\texttt{hash}_{1}$ and a collision resistant hash function $\texttt{hash}_{2}$ are involved in building the Merkle-tree.
$\texttt{hash}_{1}$ maps $\mathcal{K}$, any legal $\texttt{verify}$ module or $\texttt{info}$ into $\left\{0,1\right\}^{r}$.
$\texttt{hash}_{2}$ maps $\left\{0,1\right\}^{2r}$ into $\left\{0,1\right\}^{r}$.
The \texttt{Merkle} operator maps each component of its input by $\texttt{hash}_{1}$ then organizes the mapped values into a binary tree with $\texttt{hash}_{2}$ as the reduction operator.
Concretely, $\texttt{hash}_{1}$ for $\texttt{key}_{k}$ in the seventh and the thirteenth step in Algo.~\ref{algorithm:centralized} is instantialized as a digital signature scheme:
$$\texttt{Enc}_{k}(\texttt{hash}_{0}(\cdot)),$$
where $\texttt{Enc}_{k}$ is an encryption module using $u_{k}$'s private key, and $\texttt{hash}_{0}$ is a preimage resistant hash function, so:
\begin{itemize}
\item Each $u_{k}$ can examine whether its key has been correctly embedded into the final broadcast by reconstructing the Merkle-tree.
\item Another author $u_{s}$ cannot infer $u_{k}$'s secret key given the computational hardness of public-key encryption and the one-wayness of $\texttt{hash}_{0}$.
\item A third party can examine whether a given $\texttt{key}_{k}$ corresponds to $w=\texttt{hash}_{1}(\texttt{key}_{k})$, which are part of the evidence submitted by a author.
It decrypts $w$ using $u_{k}$'s public key and compares the plaintext with $\texttt{hash}_{0}(\texttt{key}_{k})$.
\end{itemize}
Meanwhile, $\texttt{hash}_{1}$ for the survelliance keys and all verifier modules in Algo.~\ref{algorithm:centralized} is instantialized as an encryption module using $\mathcal{A}$'s private key combined with $\texttt{hash}_{0}$.
Since it is $\mathcal{A}$'s responsibility to trace the traitor, $u_{k}$ should not be able to infer $\texttt{key}_{k}^{\dag}$, with which it can spoil the surveillance key and pirate the model.
This inference is identical to forging a new tag from a file's digital signature, which is assumed to be difficult.

\begin{figure}[htbp]
\centering
\includegraphics[width=8cm]{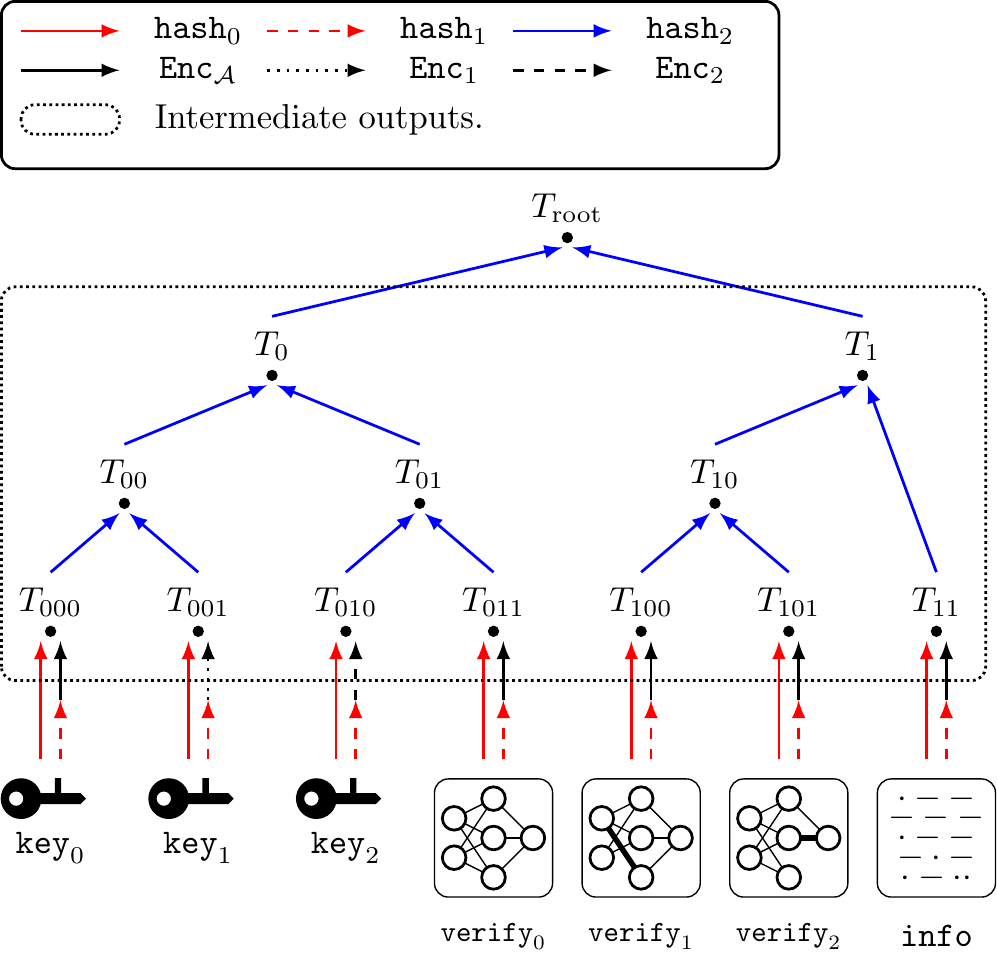}
\caption{$\texttt{Merkle}$ on three keys.}
\label{figure:4}
\end{figure}

The output of $\texttt{Merkle}$ is the value of its root node $T_{\text{root}}\in\left\{0,1\right\}^{r}$, the intermedia are the values of all the remaining nodes.
An instance for $\texttt{Merkle}$ on three keys is illustrated in Fig.~\ref{figure:4}.
The underlying setting is $K=2$ authors.
To justify the ownership, an author, e.g., $u_{1}$ submits $\texttt{key}_{1}$, $\texttt{verify}_{1}$, $\texttt{info}$ and the necessary information for deriving $T_{\text{root}}$, i.e., $T_{\text{001}}$, $T_{\text{100}}$, $T_{\text{11}}$, $T_{\text{000}}$, $T_{\text{01}}$, and $T_{\text{101}}$ to the verification community.
Any member of the community, given the public key of $u_{k}$ and the access to $\texttt{hash}_{0}$, $\texttt{hash}_{2}$ can independently examine whether the submitted evidence is consistent with the suspicious model and the root node.

This scheme is more space-friendly since the length of the root of the Merkle-tree is a constant $r$.
Moreover, the traffic burden for verification is replaced by the identity proof cost, as only an extra amount of information that is the same order as the height of the Merkle-tree, i.e., $\mathcal{O}(\ln K)$, needs to be transmitted.
If the aggregator simply hashes the keys of all authors as one string then independent verification is impossible.
While hashing and broadcasting each key multiples the communication cost by a factor $\mathcal{O}(K)$.

The basic security of this framework can be reduced to the security of the underlying watermarking scheme.
The functionality-preserving property, and the capacity of the model brace the performance of the watermarked model.
The security against tuning and piracy, and the covertness for each key ensure the robustness of each author's ownership.
We now proceed to analyze four additional requirements in the FL setting.

\subsubsection{Independency}
Since the aggregator has embedded $u_{k}$'s key set $\mathcal{K}_{k}$ into $M_{\mathcal{A}}$, transmitted $\texttt{verify}_{k}$ and the intermediate outputs of runing $\texttt{Merkle}$ to $u_{k}$, $u_{k}$ can verify its ownership over $M_{\mathcal{A}}$ independently.
To prove its ownership, $u_{k}$ only needs to submit $(\texttt{key}_{k},\texttt{verify}_{k})$ and a list of hashed strings to the public, from which $T_{\text{root}}$ can be correctly computed, and the ownership is proven.
Such verification does not involve other co-authors or the aggregator server.

\subsubsection{Privacy-preserving}
\label{section:3.3.3}
The privacy-preserving property can be formulated as the following theorem:
\begin{theorem}
\label{theorem:1}
Under the $\texttt{Merkle-Sign}$ framework, the probability that $u_{k}$ succeeds in forging $u_{s}$'s key is negligible.
\end{theorem}
\begin{proof}
We prove this statement by reduction.
If an author $u_{k}$ can use a PPT algorithm $\mathcal{A}_{\text{falsify}}$ to generate a legal key and falsify itself as another author $u_{s}$ then a PPT algorithm $\mathcal{A}_{\text{invert}}$ that inverts $\texttt{hash}_{0}$ can be built as in Algo.~\ref{algorithm:pp}.
\begin{algorithm}[htbp]
\caption{PPT $\mathcal{A}_{\text{invert}}$ that inverts $\texttt{hash}_{0}$.}
\label{algorithm:pp}
\begin{algorithmic}[1]
\REQUIRE PPT algorithm $\mathcal{A}_{\text{falsify}}$, with which $u_{k}$ can falsify itself as $u_{s}$ with non-negligible probability, $y=\texttt{hash}_{0}(x)$.
\ENSURE $\tilde{x}$ such that $\texttt{hash}_{0}(\tilde{x})=y$.
\STATE $\mathcal{A}_{\text{invert}}$ generates and distributed public and private keys for all authors.
\STATE $\mathcal{A}_{\text{invert}}$ receives the key set from $u_{k}$ running $\mathcal{A}_{\text{falsify}}$.
\STATE $\mathcal{A}_{\text{invert}}$ simulates Algo.~\ref{algorithm:centralized}, sets $\texttt{hash}_{0}(\texttt{key}_{s})$ as $y$ and builds a Merkle-tree.
\STATE $\mathcal{A}_{\text{invert}}$ runs $\mathcal{A}_{\text{falsify}}$ on the intermedia of this Merkle-tree.
\STATE $\mathcal{A}_{\text{invert}}$ returns whatever $\mathcal{A}_{\text{falsify}}$ returns.
\end{algorithmic}
\end{algorithm}

In Algo.~\ref{algorithm:pp}, the environment in which $\mathcal{A}_{\text{falsity}}$ operates is identical to that of $u_{k}$ who aims to breach the privacy of $u_{s}$.
$\mathcal{A}_{\text{invert}}$ suceeds in inverting $\texttt{hash}_{0}(x)$ iff $\mathcal{A}_{\text{falsify}}$ successfully helps $u_{k}$ in falsification, so the probabilities of both events are identical.
The preimage resistance assumption of $\texttt{hash}_{0}$ indicates that such probability is negligible, hence an effective $\mathcal{A}_{\text{falsify}}$ does not exist.
\end{proof}

\subsubsection{Recovery}
Consider the case where an adversary eavesdropping all $L$ times of $u_{k}$'s verifications and spoiling all keys.
In this case, $u_{k}$ can only prove that it has successfully inverted many nodes within the Merkle-tree, yet it is unrelated to the pirated model.
To recover its ownership, $u_{k}$ requests its neighbor w.r.t. the Merkle-tree to submit its evidence to the verification community.
For example, in Figure.~\ref{figure:4}, if $u_{1}$ has been spoiled from the model then $u_{1}$ can ask $\mathcal{A}$ to submit its key and verifier.
Given this evidence, $\mathcal{A}$ proves that it is a legal owner of the suspicious DNN model.
Then $u_{1}$'s ownership can be recovered given the fact that its evidence $(\texttt{key}_{1},\texttt{verify}_{1})$ is consistent with $\mathcal{A}$'s information within the Merkle-tree, in particular $T_{\text{001}}$ and $T_{\text{100}}$.
Therefore $u_{1}$ must have been incorporated as $\mathcal{A}$'s co-author before $\mathcal{A}$ broadcasting the message as in the thirteenth step in Algo.~\ref{algorithm:centralized}.
Each recovery involves proving one extra key that is consistent with the Merkle-tree shared by the victim's co-author, so only one co-author/aggregator is involved.
The recovery property is not contradictive to the privacy-preserving property as all keys are assumed to be one-time against potential spoil attacks.
Spoiling all keys requires an adversary to eavesdrop on all potential co-authors, which is extremely hard.
The $\texttt{Merkle-Sign}$ is secure against adversarial piracy even if recovery verification is allowed.
\begin{theorem}
\label{theorem:2}
Under the Merkle-Sign framework, the probability of pirating a published model is negligible.
\end{theorem}
This theorem can be proven by reduction as for Theorem \ref{theorem:1}.
Details are left in the appendices.

\subsubsection{Traitor-tracing}
A traitor $u_{k}$ can publish the model distributed by $\mathcal{A}$ as its product before the training terminates.
Under our configuration, $\mathcal{K}_{0}$ and $\texttt{key}_{k}^{\dag}$ are intractable to $u_{k}$ following the discussion in Section~\ref{section:3.3.3} and the privacy-preserving of the underlying watermarking scheme so $u_{k}$ cannot spoil them from the distributed model.
Hence, $\mathcal{A}$ can always successfully declare ownership over the pirated model afterward.
After doing so, $\mathcal{A}$ can examine which of the surveillance key set is in the suspicious model $M$ to locate the traitor.
For example, when $A$ finds $\texttt{verify}_{k}^{\dag}$ combined with the suspicious model correctly recognizes $\texttt{key}_{k}^{\dag}$, i.e.:
$$\text{Pr}\left\{\texttt{verify}_{k}^{\dag}(M,\texttt{key}_{k}^{\dag})=1 \right\}\geq 1-\epsilon.$$
then $u_{k}$ is the traitor.

\subsection{Generating robust triggers for $\texttt{Merkle-Sign}$}
Current black-box watermarking schemes, especially backdoor-based schemes can hardly be directly adopted in $\texttt{Merkle-Sign}$ for two reasons:
\begin{itemize}
\item The definition of $\texttt{Gen}$ is ambigious.
The keyspace and the mapping from a single key to triggers are unclear.
\item Most triggers lie in the same domain that is accessible for the adversary (e.g., triggers in~\cite{zhang2018protecting,zhu2020secure} are white noise). 
Therefore, spoiling one backdoor-back watermark brings potential harm to other watermarks. 
An adversary can spoil multiple watermarks once for all.
\end{itemize}
To cope with these problems and accommodate $\texttt{Merkle-Sign}$ to the black-box setting, we propose a new scheme for trigger generation, Autoencoder-based Trigger Generator for Federated learning ($\texttt{ATGF}$).
We resort to the autoencoder structure~\cite{li2019prove} to generate triggers, whose intermediate embedding can encode an author's key.
Since the aggregator does not possess any data in advance and lacks enough pseudorandomness to generate triggers, $\texttt{ATGF}$ absorbs patterns from authors by running the FL for an extra round as Fig.~\ref{figure:atgf}.
\begin{figure}[htbp]
\centering
\includegraphics[width=8.5cm]{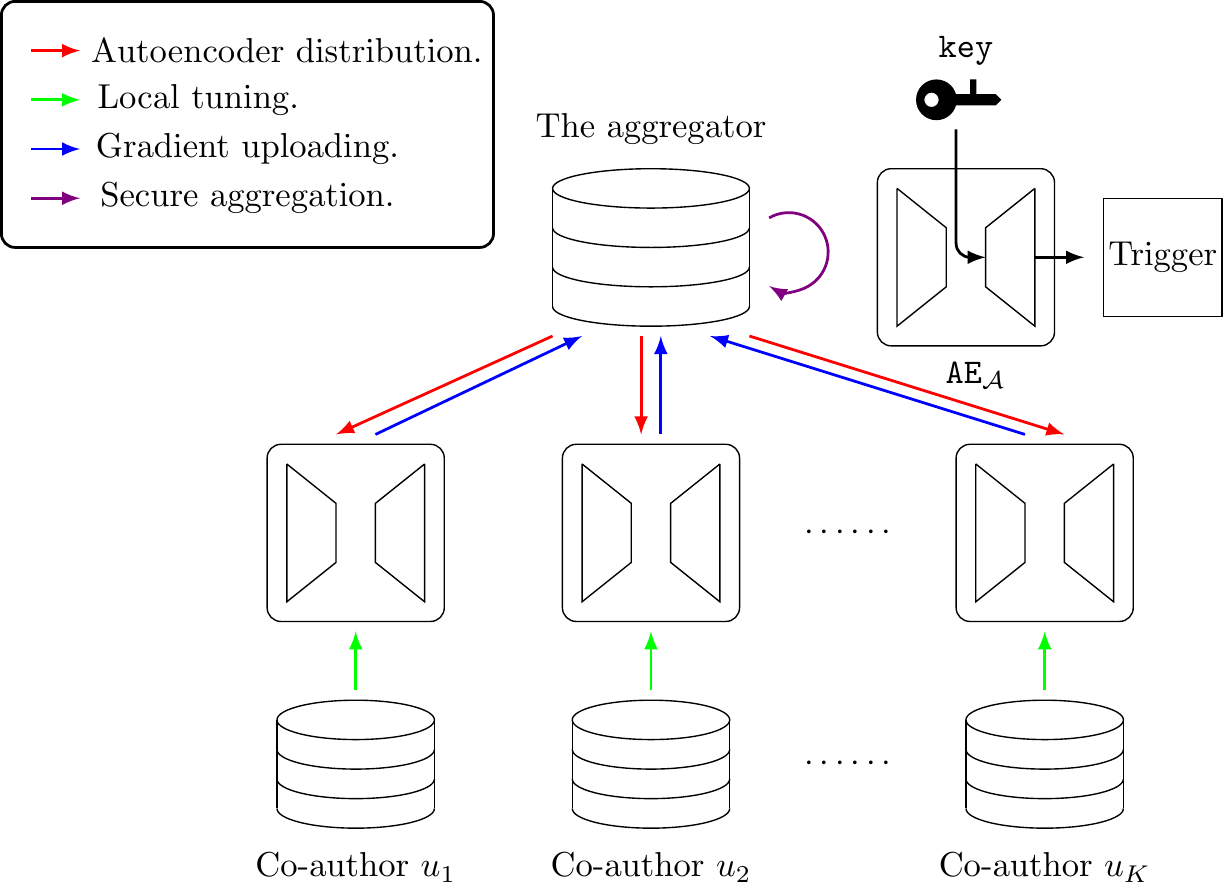}
\caption{Building the trigger generator in $\texttt{ATGF}$.}
\label{figure:atgf}
\end{figure}

Concretely, the aggregator $\mathcal{A}$ broadcasts an autoencoder structure to authors.
Each author $u_{k}$ then trains the autoencoder into $\texttt{AE}_{k}$ on its local dataset with the legitimate domain (for privacy concerns, this dataset is not necessarily that for the FL task).
Then the aggregator averages the numerical parametes across all autoencoders and obtains $\texttt{AE}_{\mathcal{A}}=\texttt{D}_{\mathcal{A}}\circ \texttt{E}_{\mathcal{A}}$.
To generate a trigger from $\texttt{key}$, we map $\texttt{key}$ into a numerical vector and feed it to $\texttt{AE}_{\mathcal{A}}$'s decoder.
The corresponding label is obtained by running another pseudorandom mapping.

The pseudorandomness within $\texttt{AE}_{\mathcal{A}}$ is attributed to all participating authors. 
An adversary oblivious to any evidence cannot infer the patterns of triggers from scratch. 
Meanwhile, since $\texttt{AE}_{\mathcal{A}}$ is an average of autoencoders trained on distinctive datasets, its decoder's output on different keys are subject to a mixed and complex distribution. 
Therefore, the spoil attack against one specific trigger in $\texttt{ATGF}$ is not going to harm other triggers significantly. 
Such decoupling between watermarks of different authors is a necessary prerequisite for the independency, recovery, and traitor-tracing properties for $\texttt{Merkle-Sign}$ against the adaptive spoil attack. 

\section{Experiment and Discussions}
\label{section:5}
\subsection{Settings}
To evaluate the performance of $\texttt{Merkle-Sign}$ and the applicability of current watermarking schemes in the FL scenario, we selected five tasks and four DNN structures for comparison.

We experimented with datasets across many domains.
Apart from MNIST~\cite{deng2012mnist} where many studies on the security on FL have been tested, we included Fashion~\cite{xiao2017/online}, a more complicated dataset, CIFAR10, and CIFAR100~\cite{krizhevsky2009learning}, two real-world large scale datasets.
Four candidate DNN structures were adopted as the backbone classifiers: an elementary multiple-layer perceptron (MLP) with cross-entropy loss, a shallow convolutional neural network (CNN) with five layers, ResNet-18, and ResNet-50~\cite{he2016deep}.
All experiments were conducted using $\texttt{PyTorch}$ framework.
In all scenarios, we adopted Adam optimizer with a three-stage learning rate decreasing schedule~\footnote{The source code will be available at \url{https://github.com/A_Temporary_User/MklSign.}}

\begin{table*}[htb]
\caption{Evaluation of watermarking capacity, the maximal number of correctly embedded watermarks when the classification error rate doubles.
$\texttt{WF}$ represents $\texttt{Wonder Filter}$, $\texttt{M-S}$ represents $\texttt{MTL-Sign}$.}
\centering
\begin{minipage}{\textwidth}
\begin{minipage}[t]{0.5\textwidth}
\centering
\begin{tabular}{m{1.3cm}<{\centering}|m{1.1cm}<{\centering}|m{1cm}<{\centering}|m{0.7cm}<{\centering}|m{0.8cm}<{\centering}|m{0.5cm}<{\centering}|m{0.9cm}<{\centering}}
\toprule
\multirow{2}{*}{\tabincell{c}{ \textbf{Dataset}}} & \multicolumn{6}{c}{MLP} \\
\cline{2-7}
& Accuracy & Uchida's & $\texttt{ATGF}$ & Rand & $\texttt{WF}$ & $\texttt{M-S}$ \\
\midrule[1pt]
MNIST & 91.92\% & 36 & 23 & 17 & 43 & 65 \\
\midrule
\tabincell{c}{Fashion} & 82.90\% & 45 & 75 & 50 & 66 & 95 \\
\midrule
CIFAR10 & 25.92\% & --- & --- & --- & --- & --- \\
\midrule
CIFAR100 & 13.26\% & --- & --- & --- & --- & --- \\
\bottomrule
\multicolumn{7}{c}{} \\

\end{tabular}
\end{minipage}
\begin{minipage}[t]{0.5\textwidth}
\centering
\begin{tabular}{m{1.3cm}<{\centering}|m{1.1cm}<{\centering}|m{1cm}<{\centering}|m{0.7cm}<{\centering}|m{0.8cm}<{\centering}|m{0.5cm}<{\centering}|m{0.9cm}<{\centering}}
\toprule
\multirow{2}{*}{\tabincell{c}{\textbf{Dataset}}} & \multicolumn{6}{c}{Shallow CNN} \\
\cline{2-7}
& Accuracy & Uchida's & $\texttt{ATGF}$ & Rand & $\texttt{WF}$ & $\texttt{M-S}$ \\
\midrule[1pt]
MNIST & 97.71\% & 110 & 12 & 13 & 17 & 55 \\
\midrule
\tabincell{c}{Fashion} & 85.40\% & 90 & 17 & 19 & 36 & 65 \\
\midrule
CIFAR10 & 61.45\% & 8 & 7 & 7 & 33 & 90 \\
\midrule
CIFAR100 & 30.03\% & --- & --- & --- & --- & --- \\
\bottomrule
\multicolumn{7}{c}{} \\
\end{tabular}
\end{minipage}

\begin{minipage}[t]{0.5\textwidth}
\centering
\begin{tabular}{m{1.3cm}<{\centering}|m{1.1cm}<{\centering}|m{1cm}<{\centering}|m{0.7cm}<{\centering}|m{0.8cm}<{\centering}|m{0.5cm}<{\centering}|m{0.9cm}<{\centering}}
\toprule
\multirow{2}{*}{\tabincell{c}{\textbf{Dataset}}} & \multicolumn{6}{c}{ResNet-18} \\
\cline{2-7}
& Accuracy & Uchida's & $\texttt{ATGF}$ & Rand & $\texttt{WF}$ & $\texttt{M-S}$ \\
\midrule[1pt]
MNIST & 99.60\% & 8,750 & 350 & 373 & 453 & 9,503 \\
\midrule
\tabincell{c}{Fashion} & 93.81\% & $\geq$10,000 & 519 & 513 & 588 & $\geq$10,000 \\
\midrule
CIFAR10 & 89.10\% & $\geq$10,000 & 582 & 572 & 663 & $\geq$10,000 \\
\midrule
CIFAR100 & 62.59\% & $\geq$10,000 & 612 & 610 & 797 & $\geq$10,000 \\
\bottomrule
\end{tabular}
\end{minipage}
\begin{minipage}[t]{0.5\textwidth}
\centering
\begin{tabular}{m{1.3cm}<{\centering}|m{1.1cm}<{\centering}|m{1cm}<{\centering}|m{0.7cm}<{\centering}|m{0.8cm}<{\centering}|m{0.5cm}<{\centering}|m{0.9cm}<{\centering}}
\toprule
\multirow{2}{*}{\tabincell{c}{\textbf{Dataset}}} & \multicolumn{6}{c}{ResNet-50} \\
\cline{2-7}
& Accuracy & Uchida's & $\texttt{ATGF}$ & Rand & $\texttt{WF}$ & $\texttt{M-S}$ \\
\midrule[1pt]
MNIST & 99.72\% & $\geq$10,000 & 417 & 411 & 494 & $\geq$10,000 \\
\midrule
\tabincell{c}{Fashion} & 95.25\% & $\geq$10,000 & 580 & 540 & 669 & $\geq$10,000 \\
\midrule
CIFAR10 & 91.50\% & $\geq$10,000 & 600 & 612 & 773 & $\geq$10,000 \\
\midrule
CIFAR100 & 67.70\% & $\geq$10,000 & 710 & 712 & 779 & $\geq$10,000 \\
\bottomrule
\end{tabular}
\end{minipage}
\end{minipage}
\label{table:5}
\end{table*}
\subsection{Revisiting watermarking schemes}
To choose the appropriate watermarking scheme for $\texttt{Merkle-Sign}$, we evaluated five candidates w.r.t. the two requirements in Section~\ref{section:2.3}, including: Uchida's~\cite{uchida2017embedding}, random trigger~\cite{zhang2018protecting,zhu2020secure}, $\texttt{Wonder}$ $\texttt{Filter}$ ($\texttt{WF}$)~\cite{li2019persistent}, $\texttt{MTL-Sign}$~\cite{ours}, and $\texttt{ATGF}$.
The scheme in~\cite{guan2020reversible} is insecure against tuning, while $\texttt{Deep-Sign}$~\cite{darvish2019deepsigns} is slow regarding watermark embedding.

In Uchida's, the key generation process selects $N_{\text{Uchida}}=20$ parameters from the DNN model architecture and $N_{\text{Uchida}}$ digits in the range $[-0.5,0.5]$ as the watermark.
Watermark embedding for Uchida's involved replacing these $N_{\text{Uchida}}$ parameters with the chosen digits.
For the two established backdoor-based schemes, the key generation process randomly selects patterns as the triggers and assigns them with randomly chosen labels as in~\cite{zhu2020secure}.
In the random trigger scheme, a random pattern with the same size as the image was generated as the key.
In $\texttt{WF}$, we generated a random stamp and set its pixel value into $\pm 2000$.
Watermark embedding is tantamount to having the DNN model learn the backdoor triggers.
For $\texttt{MTL-Sign}$, the key generation process selected a random set of $N_{\text{MTL}}=20$ digits from $[1,N_{\text{MTL}}^{3}]$, mapped them into images as QRcode, and assigned them with binary labels.
Watermark embedding for $\texttt{MTL-Sign}$ involves minimizing the binary classification loss on a watermarking backend classifier, whose inputs are the intermediate outputs of the backbone classifier.
As for $\texttt{ATGF}$, we adopted an autoencoder with $\texttt{ReLU}$ activation functions whose intermediate number of neurons are: 784, 256, 32, 256, 784.
For each author, five triggers were generated using $\texttt{ATGF}$.

\subsubsection{Capacity}
We computed the watermark capacity of the five baseline models as defined by~\eqref{equation:capacity}.
For a given model and a given dataset, the threshold of performance decline $\delta$ is set as the classification error rate of the clean model.
That is to say, if the clean model achieves classification accuracy of $(1-\delta)$ then we continue embedding watermarks into it until the accuracy declines to $(1-2\delta)$.
The results are collected in Table~\ref{table:5}, the maximal capacity was manually set to 10,000.
We observed that complicated models have a larger capacity as expected, while simple models can hardly fit real-world datasets apart from some vanilla cases.
Therefore, in the following experiments, we only adopted ResNet-50 as the backbone classifier.
The weight-based and hybrid schemes Uchida's and $\texttt{MTL-Sign}$ had a larger capacity, suitable for FL and multiple-time verification.
This is because these two methods, unlike backdoor-based methods, have a less negative impact on the DNN model.
For ResNet-18 and ResNet-50, Uchida's only modified at most 4.6\% and 1.9\% of all parameters respectively.
Meanwhile, $\texttt{MTL-Sign}$, at its extreme configuration, has almost no impact on the model to be watermarked.
Among backdoor-based schemes, $\texttt{Wonder Filter}$ had the largest capacity.
The reason behind this is that the triggers adopted in $\texttt{Wonder Filter}$ deviate from ordinary images significantly.
As a result, these triggers have less influence on the DNN's ordinary functionality.

\subsubsection{Embedding efficiency}
We recorded the time of watermark embedding for the five watermarking schemes in ResNet-50.
The time consumptions for Uchida's, random trigger, $\texttt{WF}$, $\texttt{MTL-Sign}$ and $\texttt{ATGF}$ trigger are respectively: 20.6ms, 262.2ms, 268.7ms, 617.3 ms, and 253.2ms.
Since one training epoch for ResNet-50 on average took 44s to 200s, the burden of watermark embedding in any scheme is at most upper bounded by 1.4\% of training time and is uniformly negligible.

\begin{figure}[htbp]
\subfigure[MNIST, $K=50$.]{
\begin{minipage}[htbp]{0.5\linewidth}
\includegraphics[width=4cm]{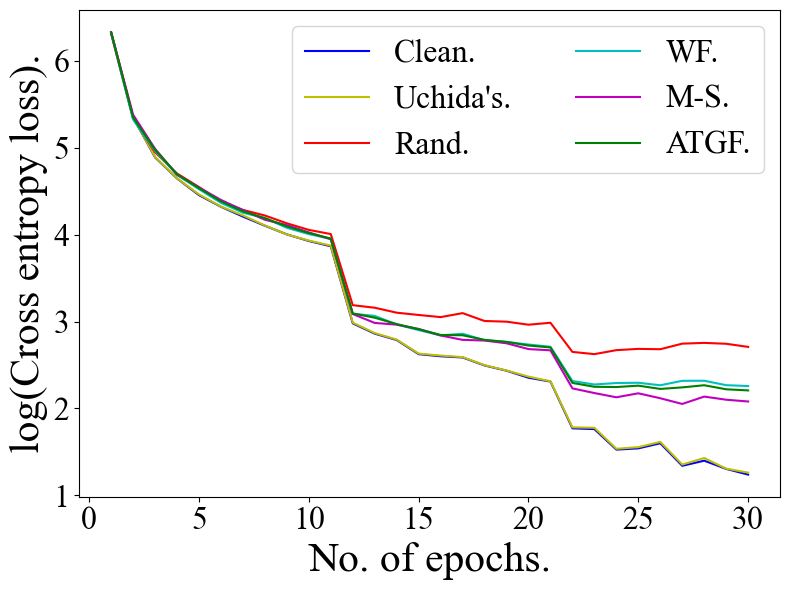}
\end{minipage}%
}%
\subfigure[CIFAR100, $K=50$.]{
\begin{minipage}[htbp]{0.5\linewidth}
\includegraphics[width=4cm]{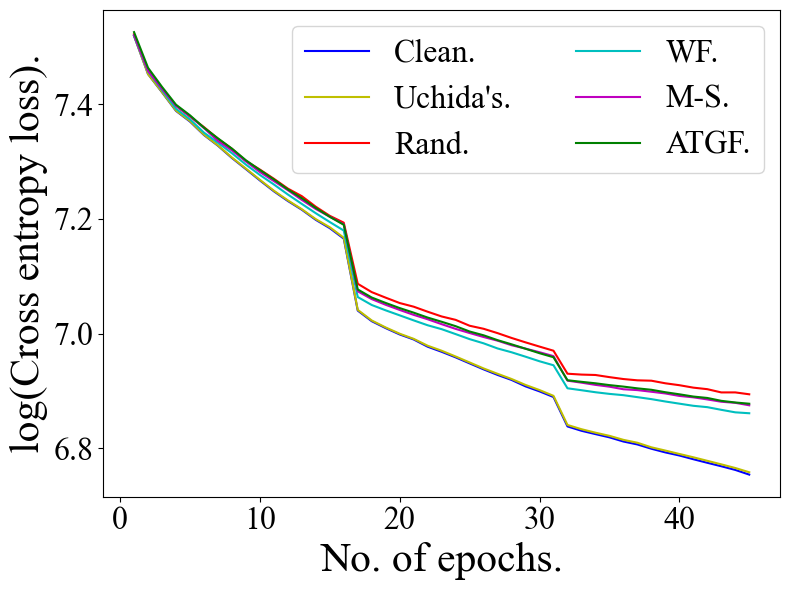}
\end{minipage}%
}%

\subfigure[MNIST, $K=100$.]{
\begin{minipage}[htbp]{0.5\linewidth}
\includegraphics[width=4cm]{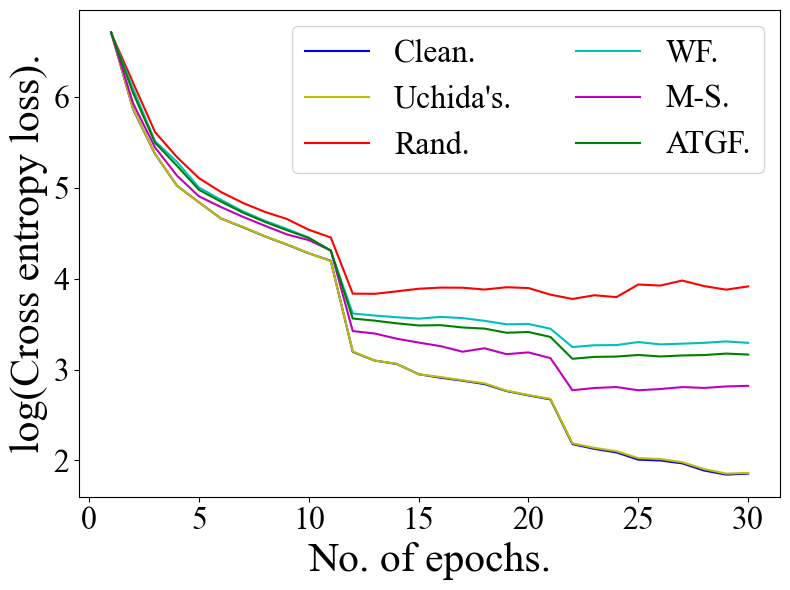}
\end{minipage}%
}%
\subfigure[CIFAR100, $K=100$.]{
\begin{minipage}[htbp]{0.5\linewidth}
\includegraphics[width=4cm]{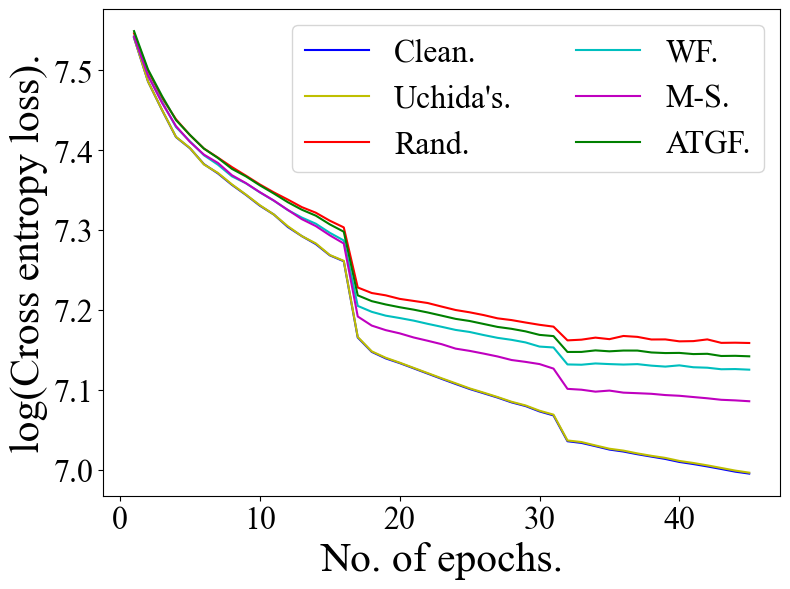}
\end{minipage}%
}%

\subfigure[MNIST, $K=200$.]{
\begin{minipage}[htbp]{0.5\linewidth}
\includegraphics[width=4cm]{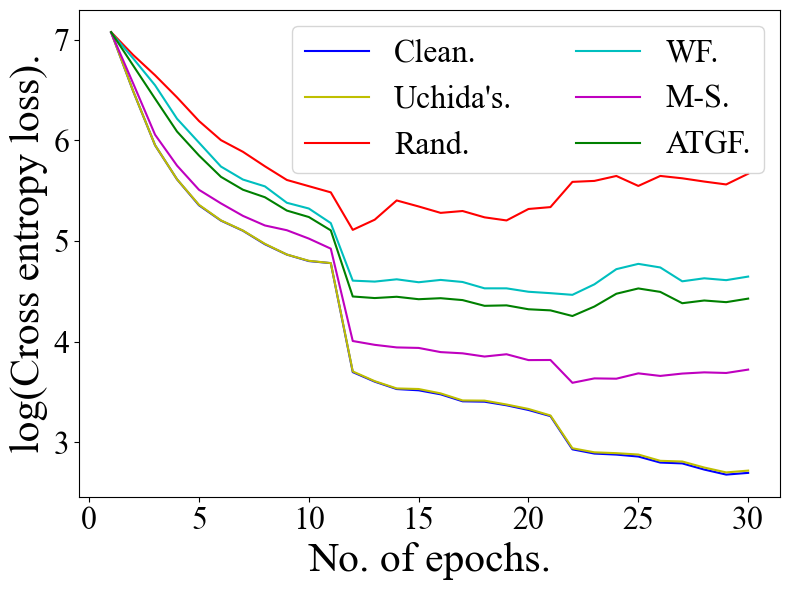}
\end{minipage}%
}%
\subfigure[CIFAR100, $K=200$.]{
\begin{minipage}[htbp]{0.5\linewidth}
\includegraphics[width=4cm]{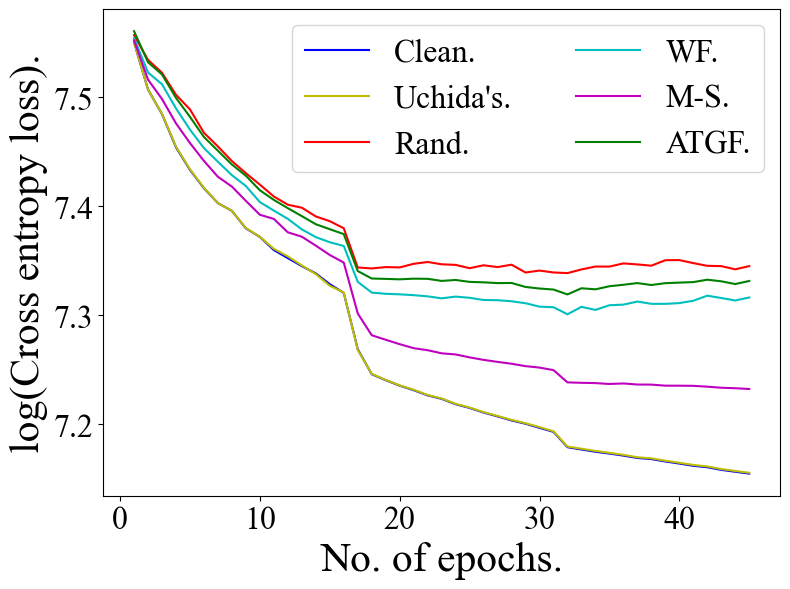}
\end{minipage}%
}%
\caption{The convergence of DNN models in FL under $\texttt{Merkle-Sign}$.}
\label{figure:dprimary}
\end{figure}

\subsection{Convergence of FL under $\texttt{Merkle-Sign}$}
Hitherto, the formal analysis in Section~\ref{section:4} has guaranteed the security of $\texttt{Merkle-Sign}$ for FL.
However, it remains unclear whether the DNN model can correctly converge or not when being watermarked intermediately.
A large capacity is necessary yet insufficient for model protection in FL.

To measure the impact of watermarking on the convergence of FL, we varied the number of authors $K$ in $\left\{50,100,200 \right\}$.
Model averaging in~\eqref{equation:ga} was adopted for aggregation in the centralized setting.
The decline of loss during the FL training is illustrated in Fig.~\ref{figure:dprimary}. 
The final classification accuracy of different configurations of $\texttt{Merkle-Sign}$ is demonstrated in Fig.~\ref{figure:table8}.
It can be observed that larger $K$ resulted in slower convergence, reflecting the trade-off between performance and security.
The scheme of Uchida's has the smallest impact on the convergence.
The impact of the other white-box scheme, $\texttt{MTL-Sign}$, is also very small.
For the three black-box/backdoor-based watermarking schemes, the convergence is at risk when $K$ approaches the watermarking capacity.
Such approaching implies that watermarking exert a large impact on the model and the normal training is disturbed.
However, in all cases, the decline of the final model's classification accuracy is upper bounded by 3.1\% according to Fig.~\ref{figure:table8} and is tolerable.
Therefore, applying $\texttt{Merkle-Sign}$ does not significantly threaten the FL system.

\begin{figure*}[htbp]
\subfigure[MNIST.]{
\begin{minipage}[htbp]{0.25\linewidth}
\includegraphics[width=4.5cm]{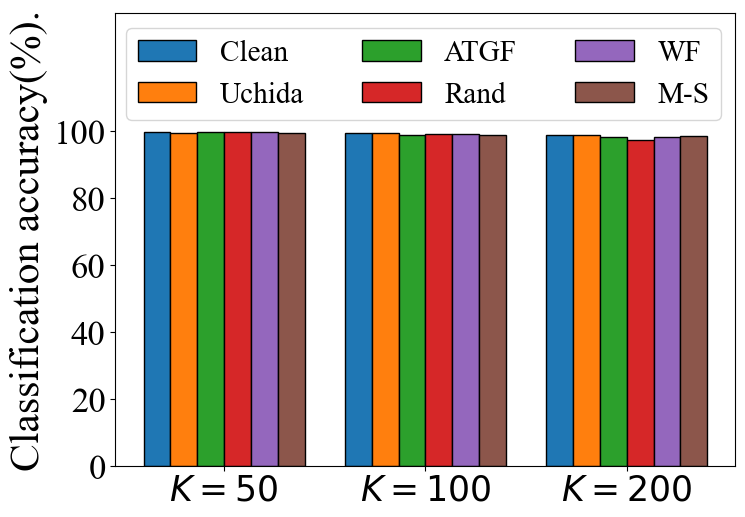}
\end{minipage}%
}%
\subfigure[Fashion.]{
\begin{minipage}[htbp]{0.25\linewidth}
\includegraphics[width=4.5cm]{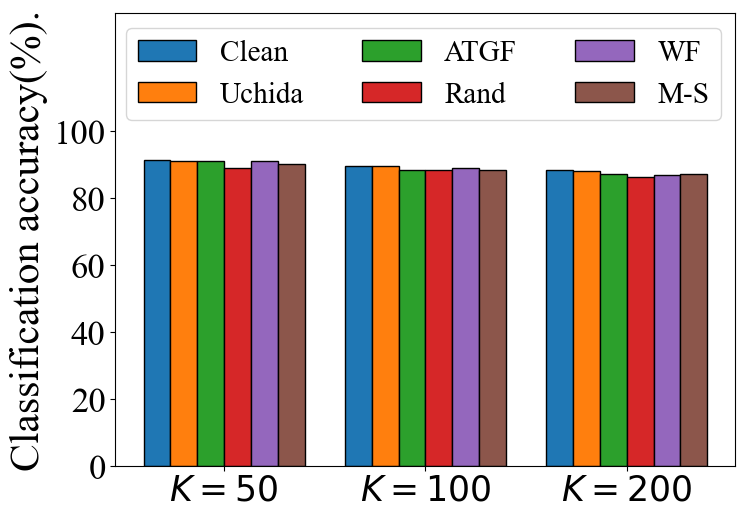}
\end{minipage}%
}%
\subfigure[CIFAR10.]{
\begin{minipage}[htbp]{0.25\linewidth}
\includegraphics[width=4.5cm]{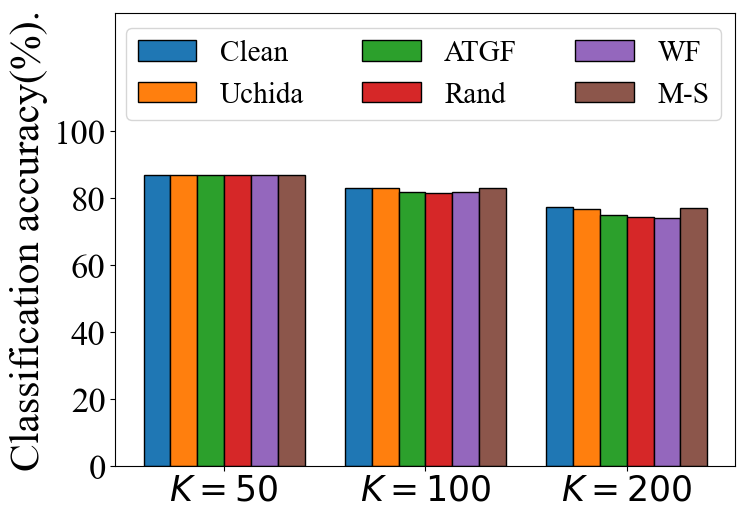}
\end{minipage}%
}%
\subfigure[CIFAR100.]{
\begin{minipage}[htbp]{0.25\linewidth}
\includegraphics[width=4.5cm]{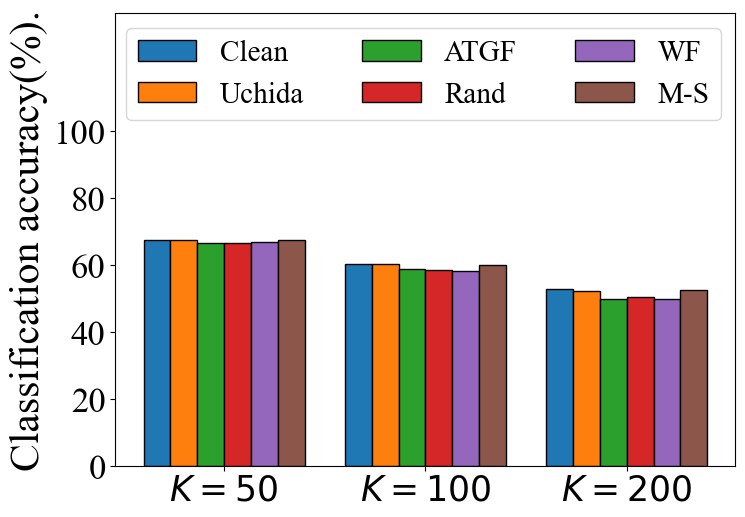}
\end{minipage}%
}%
\caption{The convergence of DNN models in FL under $\texttt{Merkle-Sign}$.}
\label{figure:table8}
\end{figure*}

\subsection{Security against the spoil attack}
Under $\texttt{Merkle-Sign}$, the cooperating authors are able to defend against the spoil attack by putting forward new evidence and reconstructing the Merkle-tree.
However, this defense might be breached due to the flaw within the watermarking scheme: the spoil attack against one watermark may erase other watermarks at the same time.
This phenomenon, especially evident in backdoor-based schemes, remains a threat to the recovery property.

To evaluate the framework's robustness against the spoil attack under different configurations, we measured the percentage of watermarks that can be correctly verified after the spoil attack against one watermark.
This metric reflects the synchronism between different watermarks within the model for a given watermarking scheme.
The lower this percentage is, the less synchronism exists between watermarks and the less effective the spoil attack is.

We adopted ResNet-50 for the evaluation, it was empirically observed that this metric was almost invariant under different datasets.
For Uchida's, the spoil attack simply replaced the watermarked parameters with random numbers.
For the spoil attack against $\texttt{MTL-Sign}$, we fixed the watermarking backend and tuned the DNN model until its watermarking branch failed to work.
For the backdoor-based schemes, we tuned the DNN model to fit arbitrary labels on the triggers.
The results are shown in Table~\ref{table:time}.
\begin{table}[htb]
\caption{Percentage of correctly verified watermarks (in \%). }
\begin{center}
\begin{tabular}{c|c|c|c|c|c}
\toprule
\multirow{2}{*}{\textbf{No. of authors}} & \multicolumn{2}{c|}{White-box scheme} & \multicolumn{3}{c}{Black-box scheme}\\
\cline{2-6}
& Uchida's &$\texttt{M-S}$ & Rand & $\texttt{WF}$ & $\texttt{ATGF}$ \\
\midrule[1pt]
{$K=50$} & 100 & 85 & 61 & 68 & 100\\
{$K=100$} & 100 & 80 & 66 & 64 & 98\\
{$K=200$} & 100 & 73 & 59 & 64 & 99\\
\bottomrule
\end{tabular}
\label{table:time}
\end{center}
\end{table}

We observe from Table~\ref{table:time} that Uchida's and $\texttt{ATGF}$ are more robust against the spoil attack since spoiling one watermark has little impact on others.
For schemes as the random trigger and $\texttt{WF}$, spoiling one single watermark can simultaneously invalidate many other watermarks.
Some instances of triggers are illustrated in Fig.~\ref{figure:triggers}. 
It can be observed that triggers generated from $\texttt{ATGF}$ were subject to a more diversified distribution. 
\begin{figure}[htbp]
\subfigure[$\texttt{ATGF}$.]{
\begin{minipage}[htbp]{0.25\linewidth}
\includegraphics[width=2.1cm]{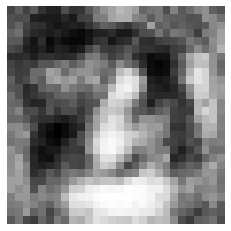}
\end{minipage}%
}%
\subfigure[$\texttt{ATGF}$.]{
\begin{minipage}[htbp]{0.25\linewidth}
\includegraphics[width=2.1cm]{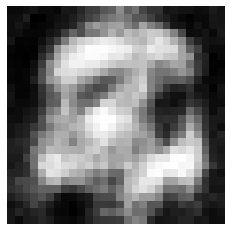}
\end{minipage}%
}%
\subfigure[$\texttt{ATGF}$.]{
\begin{minipage}[htbp]{0.25\linewidth}
\includegraphics[width=2.1cm]{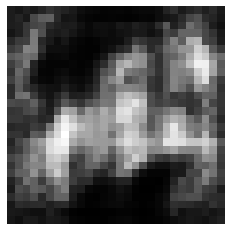}
\end{minipage}%
}%
\subfigure[$\texttt{ATGF}$.]{
\begin{minipage}[htbp]{0.25\linewidth}
\includegraphics[width=2.1cm]{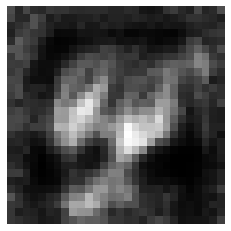}
\end{minipage}%
}%

\subfigure[Random.]{
\begin{minipage}[htbp]{0.25\linewidth}
\includegraphics[width=2.1cm]{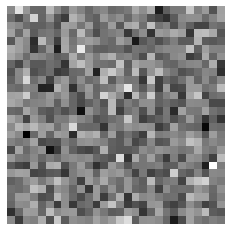}
\end{minipage}%
}%
\subfigure[Random.]{
\begin{minipage}[htbp]{0.25\linewidth}
\includegraphics[width=2.1cm]{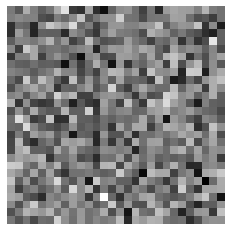}
\end{minipage}%
}%
\subfigure[$\texttt{WF}$.]{
\begin{minipage}[htbp]{0.25\linewidth}
\includegraphics[width=2.1cm]{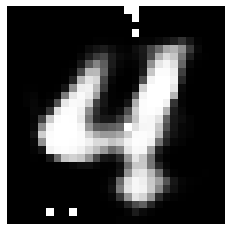}
\end{minipage}%
}%
\subfigure[$\texttt{WF}$.]{
\begin{minipage}[htbp]{0.25\linewidth}
\includegraphics[width=2.1cm]{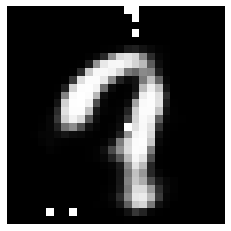}
\end{minipage}%
}%
\caption{Triggers generated from different watermarking schemes.}
\label{figure:triggers}
\end{figure}
Concretely, the spoil attack against one author in FL compromises the identity proof of 36\% and 41\% oblivious authors when $\texttt{Merkle-Sign}$ adopts these two schemes.
Under these configurations, the security of innocent authors and the recovery property are at risk.

Despite the privileges of Uchida's, it is also the scheme that is the easiest to spoil. 
The necessary time consumption for spoiling Uchida's, $\texttt{M-S}$, random trigger, $\texttt{WF}$, and $\texttt{ATGF}$ is 21ms, 750ms, 312ms, 320ms, and 303ms. 
White-box watermark schemes as Uchida's and $\texttt{MTL-Sign}$ require the author to obtain and transmit the entire suspicious DNN model to the verification community for proof.
In contrast, the black-box setting where the author only has to publicize the suspicious service deployed by the adversary.
Among compared schemes, $\texttt{ATGF}$ is the optimal choice concerning the spoil attack under the black-box setting, this is because the distribution of triggers in $\texttt{ATGF}$ is hidden from the adversary.
Therefore, we suggest using the combination of $\texttt{Merkle-Sign}$ and $\texttt{ATGF}$, under which a high level of security is guaranteed.

\section{Conclusion}
\label{section:6}
This paper presents $\texttt{Merkle-Sign}$, a framework for DNN model protection in FL by watermarking the model.
After formulating the threat model and corresponding security requirements, we extend current watermarking schemes and combine them with a data structure for distributed storage.
The security of $\texttt{Merkle-Sign}$ can be reduced to that of basic DNN watermarking schemes and cryptological primitives.
To generalize $\texttt{Merkle-Sign}$ to the black-box setting, we propose a watermarking scheme, $\texttt{ATGF}$, that generates triggers robust against the spoil attack.

Experimental results indicated that $\texttt{Merkle-Sign}$ can provide the desired security for authors in FL.
Moreover, the extra requirements introduced by $\texttt{Merkle-Sign}$ point the direction of designing new watermarking schemes.

The framework proposed in this paper sheds light on the availability of model protection in distributed learning systems such as FL.
$\texttt{Merkle-Sign}$ can be combined with other security mechanisms in FL to simultaneously protect privacy and ownership.
Our future studies are going to exploit watermarking schemes that can be more conveniently combined with distributed learning systems, e.g., watermarking schemes that are robust under the aggregation operator so the embedding can be conducted in a completely decentralized manner.

\begin{acks}
This work is supported by $\cdots$
\end{acks}

\bibliographystyle{ACM-Reference-Format}
\bibliography{WM.bib}

\appendix
\section{The Proof of Theorem \ref{theorem:2}}
\begin{proof}
To pirate a model under this verification framework, an adversary can take three possible approaches:
(1) Broadcasting a fake message with an early time-stamp than $\mathcal{A}$'s, writing its key into the model, and declaring ownership.
(2) Pretending to be an author that has participated in FL.
(3) Pretending to be an author whose key has been spoiled.

The first attack succeeds with negligible probability since the adversary has to correctly guess the structure of the DNN model to correctly provides $\texttt{info}$.
Moreover, the adversary has to tune the DNN model so its watermarking branch is consistent with its key, during which process the verifier module is fixed.
This process might bring damage to the model's performance.

For the second attack, if the adversary has not eavesdropped any verification proof then it has to find the preimage of $T_{\text{root}}$ under $\texttt{hash}_{2}$ to forge a legal key.
Hence a PPT adversary $\mathcal{A}_{\text{pretend}}$ that suceeds in pretending to be an uninformed author participating FL can be used to build a PPT adversary $\mathcal{A}_{\text{collide}}$ that finds a collision w.r.t. $\texttt{hash}_{2}$.
Formally, $\mathcal{A}_{\text{collide}}$ operates as Algo.~\ref{algorithm:rc}.
\begin{algorithm}[htbp]
\caption{PPT $\mathcal{A}_{\text{collide}}$ that finds a collsion for $\texttt{hash}_{2}$.}
\label{algorithm:rc}
\begin{algorithmic}[1]
\REQUIRE A PPT algorithm $\mathcal{A}_{\text{pretend}}$, with which an adversary can falsify itself as $u_{s}$ without eavesdropping any proof with non-negligible probability.
\ENSURE $x$ and $y$ such $\texttt{hash}_{2}(y)=\texttt{hash}_{2}(x)$.
\STATE $\mathcal{A}_{\text{invert}}$ generates and distributed public and private keys for all authors.
\STATE $\mathcal{A}_{\text{collide}}$ receives the key from the adversary runing $\mathcal{A}_{\text{pretend}}$.
\STATE $\mathcal{A}_{\text{collide}}$ simulates Algo.~\ref{algorithm:centralized}, builds a Merkle-tree with intermediate nodes $T_{0}$ and $T_{1}$.
\STATE $\mathcal{A}_{\text{collide}}$ runs $\mathcal{A}_{\text{pretend}}$ on the root node of its Merkle-tree.
\STATE $\mathcal{A}_{\text{collide}}$ receives $\mathcal{A}_{\text{pretend}}$'s output, which contains enough information for computing $T_{\text{root}}$, especially nodes of the second level $\tilde{T_{0}}$ and $\tilde{T_{1}}$.
\STATE $\mathcal{A}_{\text{collide}}$ returns $T_{0}\|T_{1}$ and $\tilde{T_{0}}\|\tilde{T_{1}}$.
\end{algorithmic}
\end{algorithm}
The event that $\mathcal{A}_{\text{collide}}$ succeeds in finding a collision pair $x\neq y$ happens if: $\mathcal{A}_{\text{pretend}}$ succeeds in pretending to be an author and $T_{0}\|T_{1}\neq\tilde{T_{0}}\|\tilde{T_{1}}$.
Therefore:
$$\text{Pr}([\mathcal{A}_{\text{collide}}\text{ wins}])\geq \text{Pr}([\mathcal{A}_{\text{pretent}}\text{ wins}]\land [T_{0}\|T_{1}\neq\tilde{T_{0}}\|\tilde{T_{1}}]).$$
Recall that the collision resistance of $\texttt{hash}_{2}$ implies that $\text{Pr}([\mathcal{A}_{\text{collide}}\text{ wins}])$ is upper bounded by a negligible function, which indicates that:
$$
\begin{aligned}
&\text{Pr}([\mathcal{A}_{\text{pretent}}\text{ wins}]\land [T_{0}\|T_{1}=\tilde{T_{0}}\|\tilde{T_{1}}])\\
=&\text{Pr}([\mathcal{A}_{\text{pretent}}\text{ wins}])-\text{Pr}([\mathcal{A}_{\text{pretent}}\text{ wins}]\land [T_{0}\|T_{1}\neq\tilde{T_{0}}\|\tilde{T_{1}}])
\end{aligned}
$$
is non-negligible given the assumption that $\text{Pr}([\mathcal{A}_{\text{pretent}}\text{ wins}])$ is non-negligible.
This further implies that $\mathcal{A}_{\text{pretend}}$ can efficiently inverse $\texttt{hash}_{2}$ with non-negligible probability, which is contradictive to the fact that $\texttt{hash}_{2}$ is also preimage resistant~\footnote{Collision resistance implies preimage resistance.}.
Therefore such $\mathcal{A}_{\text{pretend}}$ does not exist.

If the adversary has eavesdropped a verification proof then it has to find the preimage of a string under $\texttt{hash}_{1}$ for at least once.
Similar to Algo.~\ref{algorithm:pp}, the impossibility of such an attack is reduced to the security of $\texttt{hash}_{1}$.

For the third attack, the adversary has to invert $\texttt{hash}_{1}$ for at least once, which succeeds only with negligible probability as in the second case in the second attack.
\end{proof}

\section{$\texttt{Merkle-Sign}$ for Peer-to-Peer FL}
The proposed framework can be generalized to peer-to-peer (P2P), or decentralized FL, in which no aggregator is involved in the entire process. 
In P2P FL, authors form a chain along with the DNN model is trained and transmitted. 
An illustration of the P2P FL is in Fig.~\ref{figure:p2p}.
\begin{figure}[htbp]
\centering
\includegraphics[width=8cm]{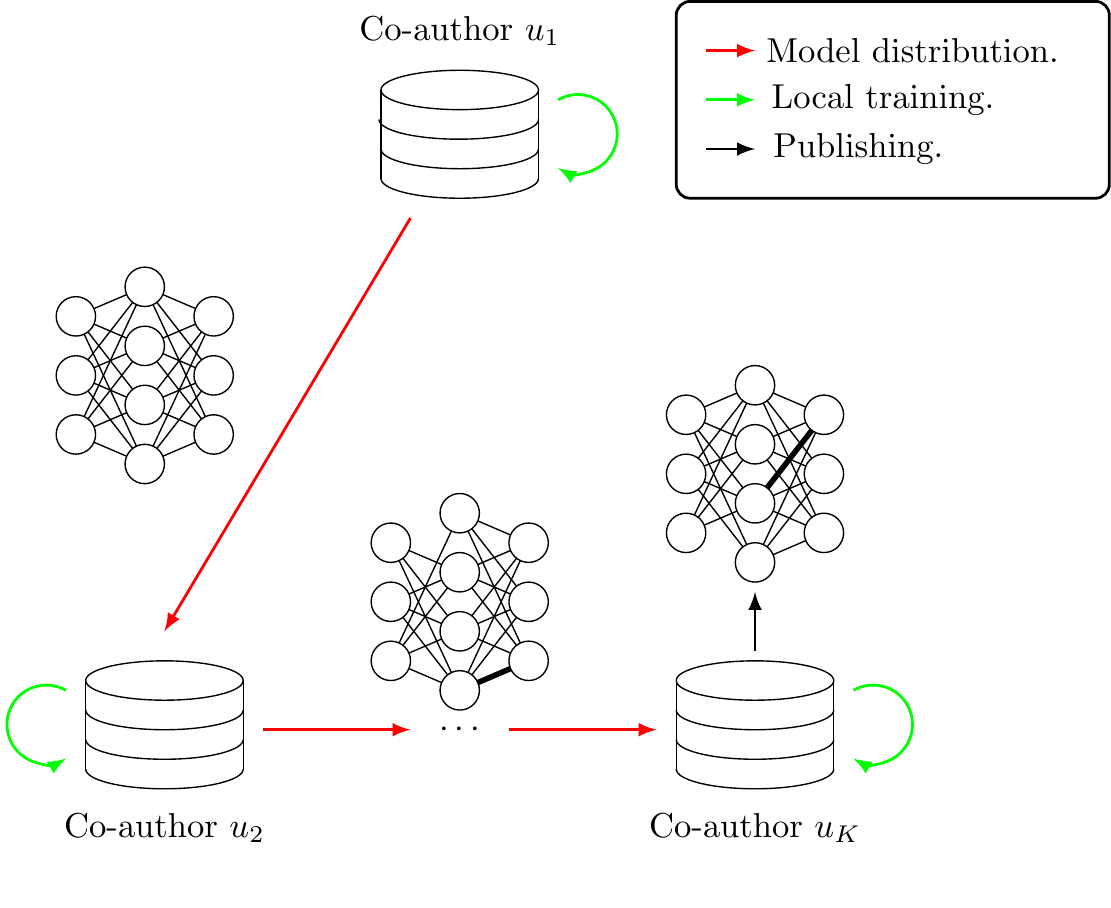}
\caption{The architectures of P2P FL.}
\label{figure:p2p}
\end{figure}

To adapt to P2P FL, $\texttt{Merkle-Sign}$ requires all authors $\mathcal{U}$ follow Algo.~\ref{algorithm:decentralized} as shown in Figure~\ref{figure:d}.
The function $\texttt{hash}_{1}$ adopted by author $u_{(t)}$ in the seventh and the eleventh steps is $\texttt{Enc}_{(t)}(\texttt{hash}_{0}(\cdot))$.
Finally, $u_{(T)}$ uses $\texttt{Enc}_{(T)}(\texttt{hash}_{0}(\cdot))$ to mask its key set and verifiers in the twelfth step.
\begin{algorithm}[htbp]
\caption{$\texttt{Merkle-Sign}$ for P2P FL.}
\label{algorithm:decentralized}
\begin{algorithmic}[1]
\REQUIRE A one-time watermarking scheme $\texttt{WM}$, security parameters $N$, $L$, functions $\texttt{hash}_{1}$ and $\texttt{hash}_{2}$.
\ENSURE A watermarked model $M_{\mathcal{A}}$ and evidence for verification.
\STATE Each $u_{k}\in\mathcal{U}$ generates $\texttt{key}_{\ k}\leftarrow\texttt{Gen}(1^{N})$.
\STATE $\mathcal{U}$ forms a schedule $(u_{(1)},u_{(2)},\cdots,u_{(T)})$.
\STATE $u_{(0)}$ initializes a clean model.
\FOR {$t=0$ to $T-1$}
\STATE $u_{(t)}$ generates time-dependent $\texttt{key}_{(t)}^{\dag}\leftarrow\texttt{Gen}(1^{N})$.
\STATE $u_{(t)}$ tunes the model, embeds $\texttt{KEYs}_{(t)}=\left\{\texttt{key}_{(t)},\texttt{key}_{(t)}^{\dag}\right\}$ into the model, obtains $\texttt{VERs}_{(t)}=\left\{\texttt{verify}_{(t)},\texttt{verify}_{(t)}^{\dag}\right\}$.
\STATE $u_{(t)}$ signs and broadcasts the following message to the verification community.
$$\langle\texttt{time}\|\texttt{Merkle}(\texttt{KEYs}_{(t)},\texttt{VERs}_{(t)},\texttt{info}) \rangle.$$
\STATE $u_{(t)}$ transmits the model to the $u_{(t+1)}$.
\ENDFOR
\STATE $u_{(T)}$ tunes and embeds $\mathcal{K}_{(T)}$ into the model, obtains $\left\{ \texttt{verify}_{(T)}\right\}$.
\STATE Each $u_{k}\in\mathcal{U}/\left\{u_{(T)}\right\}$ sends $\left\{\texttt{hash}_{1}(\texttt{key}_{k})\right\}$ and $\left\{\texttt{hash}_{1}(\texttt{verify}_{k})\right\}$ to $u_{(T)}$.
\STATE $u_{(T)}$ builds a Merkle-tree from the information received, signs and broadcasts:
$$\langle \texttt{time}\|\texttt{Merkle}(\texttt{KEYs},\texttt{VERs},\texttt{info}) \rangle.$$
\STATE $u_{(T)}$ publishes the model.
\end{algorithmic}
\end{algorithm}

We assume that the watermarking scheme only modifies a tiny part of the model (e.g, few parameters or few trigger samples) so such embedding is not going to prevent the model from convergence.
The capacity of the model w.r.t. the watermarking scheme has to be larger than $T$ to ensure the model's performance.
The basic security requirements, the independency, the privacy-preserving, and the recovery properties hold with similar discussion as in the aggregator-based case.
\begin{figure}[htbp]
\centering
\includegraphics[width=8cm]{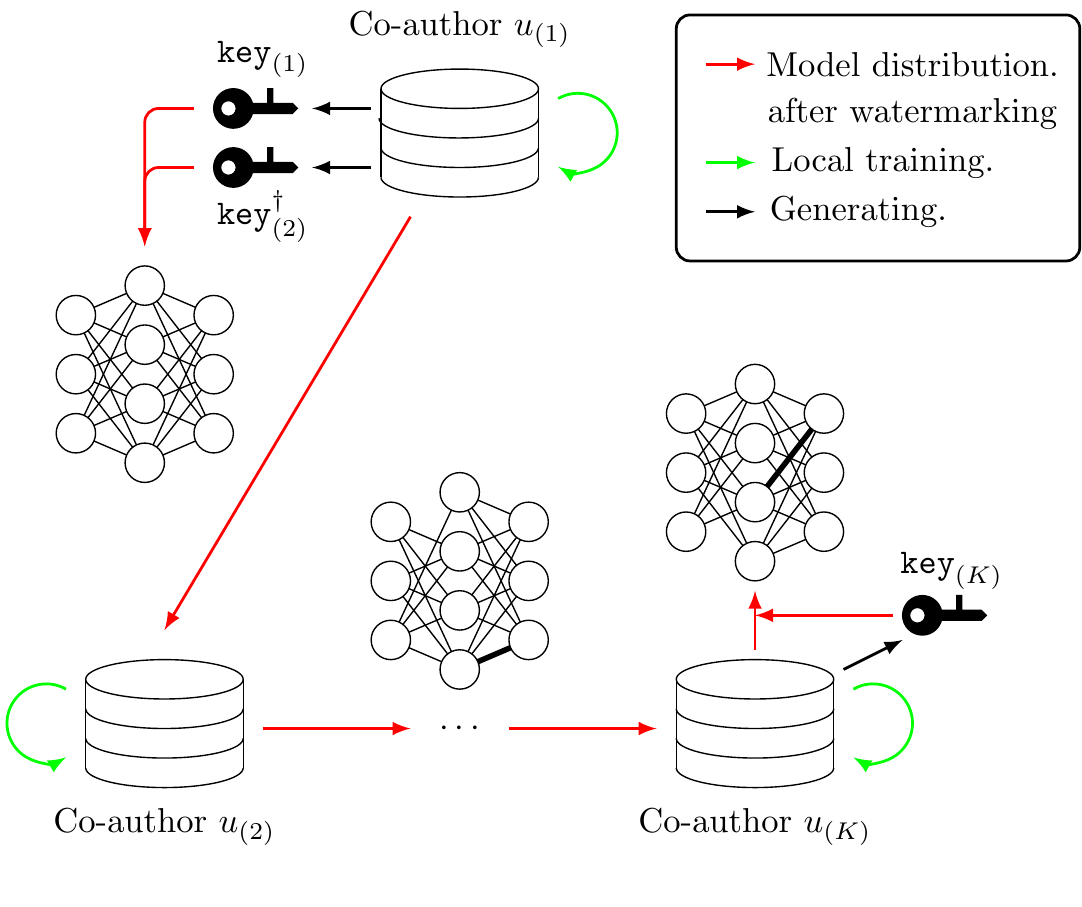}
\caption{The $\texttt{Merkle-Sign}$ watermarking framework for P2P FL.}
\label{figure:d}
\end{figure}
For traitor-tracing, the authors $\mathcal{U}$ follow Algo.~\ref{algorithm:ttd}.
The correctness is guaranteed by the following theorem.
\begin{algorithm}[htbp]
\caption{Traitor-tracing under $\texttt{Merkle-Sign}$ for P2P FL.}
\label{algorithm:ttd}
\begin{algorithmic}[1]
\REQUIRE The pirated model $M$.
\ENSURE The index of the traitor.
\FOR {$t=1$ to $T$}
\STATE $u_{(t)}$ proves its ownership over $M$.
\STATE $u_{(t)}$ presents to its co-authors: $\texttt{key}_{(t)}^{\dag}$, $\texttt{verify}_{(t)}^{\dag}$, and $\left\{\texttt{hash}_{1}(\texttt{key}_{(t)}), \texttt{hash}_{1}(\texttt{verify}_{(t)})\right\}$.
\STATE Other authors check whether the presented information is consistent with $u_{(t)}$'s historical broadcasting.
\STATE Other authors vote on $\texttt{verify}_{(t)}^{\dag}\left(M,\texttt{key}_{(t)}^{\dag}\right)$.
\STATE If the voted result is zero then return $(t-1)$.
\ENDFOR
\STATE Return $(T)$.
\end{algorithmic}
\label{exp:1}
\end{algorithm}
\begin{theorem}
If $u_{(t)}$ is the traitor who sells the model, then Algo.~\ref{algorithm:ttd} can almost always correctly identify it.
\end{theorem}
\begin{proof}
We assume that $u_{(t)}$ follows Algo.~\ref{algorithm:decentralized} except for its unauthorized reselling.
Otherwise, it cannot verify and protect its contribution in the model, and it might be identified as the traitor from the fourth step in Algo.~\ref{algorithm:ttd}.
Compared with the model released by $u_{(t)}$ to $u_{(t+1)}$, the model in $u_{(t-1)}$'s perspective does not contain $\texttt{key}_{(t)}^{\dag}$.
The secrecy of $\texttt{key}_{(t)}^{\dag}$ is protected by the CPA-security, the covertness, and the privacy-preserving property of the underlying watermarking scheme, together with the one-wayness of $\texttt{Merkle}$.
If $u_{(t-1)}$ can successfully embed this surveillance key onto the model and escape tracing then it must have obtained $\texttt{key}_{(t)}^{\dag}$, which happens only with negligible probability.
By the definition of the correctness of $\texttt{WM}$, $u_{(t-1)}$ cannot forge the evidence for $\texttt{key}_{(t)}^{\dag}$.
Therefore, Algo.~\ref{algorithm:ttd} is sufficient to identify $u_{(t-1)}$ as the traitor.
\end{proof}

\end{document}